\documentclass[11pt]{article}
\usepackage[latin1]{inputenc}        
\usepackage[T1]{fontenc}             
\usepackage[intlimits]{amsmath}      
\usepackage{amsfonts}                
\usepackage{amssymb}                 
\usepackage{amsthm}                 
\usepackage{textcomp}                
\usepackage{array}                   
\usepackage{enumerate}               
\usepackage{titlesec}
\usepackage{Tabbing}
\usepackage{graphicx}
\usepackage{geometry} 
\usepackage{longtable} 
\usepackage{setspace}
\usepackage{multirow}
\usepackage{todo}

\usepackage{tikz}
\usepackage{verbatim}
\usetikzlibrary{arrows,shapes}

\usepackage{cite} 

\newtheorem{defi}{Definition}

\newtheorem{prop}{Proposition}
\newtheorem{satz}{Theorem}
\newtheorem*{bem}{Remark}
\newtheorem{fol}{Corollary}

\newtheoremstyle{mytheoremstyle2} 
    {\topsep}                    
    {\topsep}                    
    {}                           
    {}                           
    {\bfseries}                   
    {:}                           
    {.5em}                       
    {}  
\newtheorem*{bspi}{Example}

\newcommand{\classP}{{\sf P}}
\newcommand{\classNP}{{\sf NP}}
\newcommand{\classAPX}{{\sf APX}}

\newcommand{\Reach}{{\sc Path-To-Stability}}
\newcommand{\ThreeSAT}{{\sc 3Sat}}

\newcommand{\SRT}{{\sc SRT}}
\newcommand{\LocalSRT}{{\sc Local-SRT}}
\newcommand{\SMTI}{{\sc SMTI}}

\begin{document}

\tikzstyle{existEdge} = [draw,line width=3pt,-]
\tikzstyle{link} = [draw,dashed,-]
\tikzstyle{edge} = [draw,-]

\title{Locally Stable Marriage with Strict Preferences\thanks{An extended abstract of this paper has been published in the proceedings of the 40th Intl.\ Colloquium on Automata, Languages, and Programming (ICALP 2013)~\cite{HoeferW13}.}}
\author{Martin Hoefer\thanks{Max-Planck-Institut f\"ur Informatik and Saarland University, Germany. {\tt mhoefer@mpi-inf.mpg.de}. Supported by DFG Cluster of Excellence MMCI.} \and Lisa Wagner\thanks{Department of Computer Science, RWTH Aachen University, Germany. {\tt lwagner@cs.rwth-aachen.de}. Supported by DFG grant Ho 3831/3-1}}
\date{}

\maketitle

\begin{abstract}
We study stable matching problems with locality of information and control. In our model, each agent is a node in a fixed network and strives to be matched to another agent. An agent has a complete preference list over all other agents it can be matched with. Agents can match arbitrarily, and they learn about possible partners dynamically based on their current neighborhood. We consider convergence of dynamics to locally stable matchings -- states that are stable with respect to their imposed information structure in the network. In the two-sided case of stable marriage in which existence is guaranteed, we show that the existence of a path to stability becomes \classNP-hard to decide. This holds even when the network exists only among one partition of agents. In contrast, if one partition has no network and agents remember a previous match every round, a path to stability is guaranteed and random dynamics converge with probability 1. We characterize this positive result in various ways. For instance, it holds for random memory and for cache memory with the most recent partner, but not for cache memory with the best partner. Also, it is crucial which partition of the agents has memory. Finally, we present results for centralized computation of locally stable matchings, i.e., computing maximum locally stable matchings in the two-sided case and deciding existence in the roommates case.
\end{abstract}

\thispagestyle{empty}
\clearpage


\section{Introduction}
\label{sec:intro}
Matching problems form the basis for a variety of assignment and allocation tasks encountered in computer science, operations research, and economics. A prominent and popular approach in all these areas is \emph{stable matching}, as it captures aspects like distributed control and rationality of participants that arise in many assignment problems today. A variety of allocation problems in markets can be analyzed within the context of two-sided stable matching, e.g., the assignment of jobs to workers~\cite{KelsoC82,ArcauteV09}, organs to patients~\cite{RothSU05}, or general buyers to sellers. In addition, stable marriage problems have been successfully used to study distributed resource allocation problems in networks~\cite{AkkayaGB10,Mathieu08,GoemansLMT06}.

In this paper, we consider a game-theoretic model for decentralized matching with limited information. A set of rational agents are embedded in a (social) network and strive to find a partner for a joint relationship or activity, e.g., to do sports, write a research paper, exchange data etc. Such problems are of central interest in economics and sociology, and they act as fundamental coordination tasks in distributed computer networks. Our model extends the stable marriage problem, in which we have sets $A$ and $B$ of men and women. Each man (woman) can match to at most one woman (man) and has a complete preference list over all women (men). Each agent would rather be matched than unmatched. Given a matching $M$, a \emph{blocking pair} is a man-woman pair such that both improve by matching to each other and leaving their current partner (if any). A matching without any blocking pair is a \emph{stable matching}. 

A central assumption in stable marriage is that every agent knows all agents it can match to. In reality, however, agents often have limited information about their matching possibilities. For instance, in a large society we would not expect a man to match up with any other woman immediately. Instead, there exist restrictions in terms of knowledge and information that allow some pairs to match up directly, while others would have to get to know each other first before being able to start a relationship. Similar information restrictions exist in many large matching markets -- in job markets, one of the most successful strategies to find a new job is to rely on personal contacts that allow the discovery of job opportunities before they become public knowledge. In housing markets, a popular strategy for finding a new roommate is to rely on personal contacts in the social network to find possible tenants. In dating markets, agents usually get to know only a small subset of agents based on which further contacts can evolve. The same happens in many matching scenarios, e.g., for finding partners to travel or to do sports.

We incorporate this aspect of local information by assuming that agents are embedded in a fixed network of \emph{links} $L$. Links represent an enduring knowledge relation that is not primarily under the control of the agents. Depending on the interpretation, links could represent, e.g., family, neighbor, co-worker or teammate relations. Each agent strives to build one \emph{matching edge} to a partner. The set of links and edges defines a dynamic information structure based on \emph{triadic closure}, a standard idea in social network theory: If two agents have a common friend, they are likely to meet and learn about each other. Translated into our model this implies that each agent can match to partners in its 2-hop neighborhood of the network of matching edges and links. To clarify our discussion, we present a formal definition here. For more preliminaries and details on the model see Section~\ref{sec:model}.

\begin{defi}[Locally Stable Matching]
We denote by $V$ the set of agents, and $E \subset V \times V$ the set of possible matching edges. A \emph{local blocking pair} of matching $M$ is a blocking pair $\{u,v\} \in E$ of agents $u,v \in V$ which are at hop distance at most 2 in the network $G=(V,M\cup L)$. Consequently, a \emph{locally stable matching} is a stable matching without local blocking pairs.
\end{defi}

Local blocking pairs are a subset of blocking pairs. This implies that every stable matching is a locally stable matching, because it allows no (local or global) blocking pairs. Thus, one might be tempted to think that locally stable matchings are easier to find and/or reach using distributed dynamics than ordinary stable matchings. In contrast, we show in this paper that locally stable matchings have a rich structure and can behave quite differently than ordinary stable matchings. Our study of locally stable matching with general strict preferences significantly extends recent work on the special case of correlated preferences~\cite{Hoefer13,HoeferW14,HoeferVW15}, in which preferences are determined by benefits for each matching edge. 

For most of the paper, we concentrate on the important two-sided scenario of stable marriage, in which a (locally) stable matching is always guaranteed to exist. Our primary interest is to characterize convergence properties of iterative round-based dynamics with distributed control, in which in each round a local blocking pair is resolved. We focus on the \Reach\ problem:

\begin{defi}[\Reach]
Given a local matching game and an initial matching $M_0$, the problem is to decide if there is a path to stability, i.e., a sequence of local blocking pair resolutions leading from $M_0$ to an arbitrary locally stable matching. As a variation, we sometimes consider existence of a path to a \emph{given} locally stable matching instead of an arbitrary one.
\end{defi}

The properties of such dynamics in large matching markets are a broad domain of interest in the literature~\cite{PaisPV12,EcheniqueY12} (see also our discussion of related work below). We show that, in contrast to stable matching without locality restrictions, \Reach\ can be \classNP-hard, even in surprisingly special cases. There are two conditions we identify to overcome this lower bound -- correlated structure in the preference lists, and memory.

Memory is a natural and important characteristic of human behavior. Moreover, even when agents are, e.g., organizations, or software/hardware components, it is reasonable to assume a certain form of memory. We take a quite conservative approach to this issue and assume that in each step of a sequence, each agent can remember only a single other agent that it was matched to before. This single agent can be chosen in various ways, and we consider three natural strategies -- the best partner (quality memory), the most recent partner (recency memory), a random partner (random memory). Our results show that the right kind of memory can indeed drastically change the properties of \Reach. In fact, perhaps the most effective strategy of the three is random memory, since it allows to guarantee convergence in the limit with probability 1. However, recency memory also can be effective under some restrictions on the structure of the instance. Quality memory, however, does not seem to significantly change the complexity of \Reach.

\subsection{Results and Contribution}

In Section~\ref{sec:reach} we first derive an example instance, where a locally stable matching can never be reached when starting from the empty matching. This is in strong contrast to the case of correlated preferences (sometimes referred to as weighted matching, or globally ranked pairs), in which it is easy to show convergence of every sequence of local blocking pair resolutions with a potential function~\cite{Mathieu08}. In fact, we use this gadget to show that it is \classNP-hard to decide \Reach, even if there are no links within one partition of agents. If we need to decide \Reach\ from a given initial matching \emph{to a given locally stable matching}, we show that this is \classNP-hard even for correlated preferences, and even if the links are only within one partition of agents. 

Moreover, we prove that there exist games and initial matchings such that \emph{every} sequence of local blocking pair resolutions terminating in a locally stable matching is exponentially long. Hence, in general \Reach\ might even be outside \classNP. In contrast, for correlated preferences, we show that every reachable state can also be reached using a sequence of polynomial length. Thus, for correlated preferences, \Reach\ is in \classNP\ and, thus, \classNP-complete if we ask for a given locally stable matching\footnote{Note that it is always true if we ask for an arbitrary one~\cite{Hoefer13}.}.

Given our \classNP-hardness results, in Section~\ref{sec:memory} we concentrate on a more restricted class of games in which one partition has no internal links, i.e., links exist only between partitions and among the other partition. This is a natural assumption when considering objects that do not generate knowledge about each other, e.g., when matching resources to networked nodes or users, where initially resources are only known to a subset of users. Here we characterize the impact of memory on distributed dynamics. For \emph{recency memory}, each agent remembers in every round the \emph{most recent partner} that is different from the current one. With recency memory, we show that \Reach\ is always true, and for every initial matching there exists a sequence of polynomially many local or remembered blocking pairs leading to a locally stable matching. In fact, we only need the partition without internal links to have recency memory. If, in contrast, only the other partition has recency memory, \Reach\ becomes again \classNP-hard. The same hardness holds for \emph{quality memory} if all agents from both partitions remember their \emph{best partner}. Our results formally support the intuition that recency memory is more powerful than quality memory, as the latter can be easily misled in the course of a dynamic process. This provides a novel distinction between recency and quality memory that was not known in previous work~\cite{Hoefer13}.

Our positive results for recency memory in Section~\ref{sec:memory} imply that if we pick admissible blocking pairs uniformly at random in each step, we achieve convergence with probability 1. This can also be guaranteed for \emph{random memory} if in each round each agent remembers one of his previous matches chosen uniformly at random. In fact, for random memory this result holds even in general when links exist among or between both partitions. However, using known results on stable marriage with full information~\cite{AckermannGMRV11}, convergence time can be exponential with high probability, independently of any memory.

In Sections~\ref{sec:IS} and~\ref{sec:roommate} we treat more centralized aspects of locally stable matching to highlight their different nature compared to ordinary stable matchings. A fundamental observation that motivates our results in Section~\ref{sec:IS} is that -- in contrast to ordinary stable matchings -- two locally stable matchings can have different sizes, and we consider the natural problem of finding a locally stable matching of maximum cardinality. This problem is known to be \classAPX-hard~\cite{ChengM13}. While a simple 2-approximation algorithm exists, we can show a non-approximability result of $1.5-\varepsilon$ under the unique games conjecture. Finally, in Section~\ref{sec:roommate} we consider the roommates problem, in which agents can match arbitrarily to other agents. In this case, we show that -- in contrast to ordinary stable matchings -- deciding existence of locally stable matchings is \classNP-complete.

Note that we prove all our hardness results for the case when agents have complete lists. In contrast, all our positive results are shown for the case of incomplete lists, where for each agent the set of possible matching partners is restricted to some arbitrary subset of agents (or, in case of stable marriage, subset of the other partition).

\subsection{Related Work}
Locally stable matchings were introduced by Arcaute and Vassilvitskii~\cite{ArcauteV09} in a two-sided job-market model, in which links exist only among one partition. The paper uses strong uniformity assumptions on the preference lists and addresses the lattice structure for stable matchings and a local Gale-Shapley algorithm. 

More recently, we studied locally stable matching and extensions with correlated preferences~\cite{Hoefer13}. In the roommates problem, where arbitrary pairs of agents can be matched, a potential function argument shows that \Reach\ is always true and convergence guaranteed~\cite{AbrahamIKM07,Mathieu10}. Moreover, in~\cite{Hoefer13} we proved that for every initial matching there is a polynomial path to stability, i.e., a sequence of local blocking pairs that leads to a locally stable matching. The expected convergence time of random dynamics, however, can be exponential. If we restrict to resolution of pairs with maximum benefit, then for random memory the expected convergence time becomes polynomial, but for recency or quality memory convergence time remains exponential, even if the memory is of polynomial size. 

Subsequent to publication of the extended abstract of this paper in~\cite{HoeferW13}, we studied more general coalition formation games with correlated preferences in which a polynomial path to stability can always be guaranteed. In~\cite{HoeferW14} we extended some of the results shown here (the \classNP-hardness result in Theorem~\ref{satz1} and the existence result in Theorem~\ref{satz3}) to other variants of matching games, such as considerate or friendship matching~\cite{AnshelevichBH13} that capture externalities among agents. In addition, we provided a tight characterization for a class of games with more general dynamic restrictions on the set of available blocking coalitions. This class contains locally stable matching as well as a variety of other variants of matching games. In~\cite{HoeferVW15} we showed a tight characterization for graph-based coalition formation games with limited visibility. This class extends the ideas underlying locally stable matching to coalitions of larger size.

For ordinary two-sided stable matching there have been a wide variety of works on various aspects, e.g., many-to-many matchings, ties, incomplete lists, etc. For an introduction to the topic we refer the reader to several books in the area~\cite{GusfieldI89,RothS90,Manlove13}. Theoretical work on convergence issues in ordinary stable marriage has focused on better-response dynamics, in which agents sequentially deviate to blocking pairs. It is known that for stable marriage these dynamics can cycle~\cite{Knuth76}. On the other hand, \Reach\ is always true, and for every initial matching there exists a polynomial path to stability~to an arbitrary stable matching~\cite{RothVV90}. Our recent work~\cite{HoeferW14} shows that deciding \Reach\ to a \emph{given} stable matching becomes \classNP-hard. If blocking pairs are chosen uniformly at random at each step, convergence time to an arbitrary stable matching can be exponential~\cite{AckermannGMRV11}. More recently, several works studied convergence time of random dynamics using combinatorial properties of preferences~\cite{HoffmanMP13}, or the probabilities of reaching certain stable matchings via random dynamics~\cite{BiroN13}. A prominent open problem in this domain is deciding \Reach\ for a sequence of blocking switches, where if a blocking pair $\{a,b\}$ deviates, then their former partners also pair up immediately. A path to stability does not necessarily exist \cite{Tamura93, TanS95}, but up to our knowledge the complexity of deciding \Reach\ is still open.

For the more general roommates problem, in which every pair of agents can be matched, stable matchings can be absent. There are algorithms to decide existence and compute stable matchings in polynomial time if they exist~\cite{Irving85}. Some of these algorithms rely on a tight characterization of stable matchings in terms of linear programming~\cite{Teo98}. In terms of convergence, if a stable matching exists, there also exist polynomial paths to stability to a stable matching from every initial matching~\cite{DiamantoudiMX04}. Similar results hold for more general concepts like $P$-stable matchings that always exist~\cite{InarraLM08}. Ergodic sets of the underlying Markov chain have been studied~\cite{InarraLM10} and related to random dynamics~\cite{KlausFW10}. Alternatively, several works have studied the computation of (variants of) stable matchings using iterative entry dynamics~\cite{BlumRR97,BlumR02,BiroCF08,Cheng16}.

Recently, computing locally stable matchings with maximum cardinality has attracted some interest. While local algorithms perform arbitrarily badly~\cite{BrandesW12}, the problem was shown to be \classAPX-hard in~\cite{ChengM13} and non-approximable within $21/19$ unless $\classP \neq \classNP$. We show a stronger lower bound of $1.5-\varepsilon$ under unique games conjecture. The same question has been studied for a related variant termed socially stable matching. Askalidis et al~\cite{AskalidisIKMP13} show that the problem in this variant is also non-approximable within $1.5-\varepsilon$ under the unique games conjecture by adapting the proof technique of this paper. They also provide a 1.5-approximation algorithm for socially stable matchings.

\section{Preliminaries}
\label{sec:prelim}

\subsection{Locally Stable Matching}
\label{sec:model}

A \emph{network matching game} (or \emph{network game}) consists of a (social) \emph{network} $N=(V,L)$, where $V$ is a set of vertices representing \emph{agents} and $L\subseteq \{\{u,v\}\mid u,v\in V, u\neq v\}$ is a set of fixed \emph{links}. The set $E\subseteq \{\{u,v\}\mid u,v\in V, u\neq v\}$ defines the \emph{potential matching edges}. Note that, in general, $E$ is not necessarily a sub- or superset of $L$. A \emph{state} $M\subseteq E$ is a matching, where for each $v\in V$ we have $|\{e\mid e\in M, v\in e\}| \leq 1$. Based on a state, we define the \emph{graph} to be $G_M = (V,L\cup M)$.

If an edge $e=\{u,v\}\in M$, the incident agents obtain a utility of $b_u(e), b_v(e) > 0$ for $u$ and $v$, respectively. Alternatively, we often represent the preference of each agent $v$ by a list $\succ_v$ over all its possible matching partners, where $u \succ_v w$ if $b_v(\{v,u\}) > b_v(\{v,w\})$ and $u =_v w$ if $b_v(\{v,u\}) =_v b_v(\{v,w\})$, for any $\{v,u\}, \{v,w\} \in E$. Note here that our model includes instances with incomplete lists and ties, and all our upper bounds hold for this general case. In contrast, our lower bounds apply even for the case of strict preference lists. As a special case, we also study games with \emph{correlated preferences} (also termed \emph{correlated network games}), in which it holds that $b_u(e) = b_v(e) = b(e) > 0$ for every $e \in E$. 

If we can divide $V$ into two disjoint sets $U$ and $W$ such that $E\subseteq \{\{u,w\}\mid u\in U, w\in W\}$, we call the game \emph{bipartite}. We will focus on this case in Section~\ref{sec:reach} and Section~\ref{sec:memory}. Note that this does not imply that $N$ has to be bipartite. If further the agents of $U$ are isolated in $N$, we term the game a \emph{job-market game} for consistency with~\cite{ArcauteV09,Hoefer13}.

To describe stability in network matching games, we extend the classic notion of blocking pair. Consider a matching $M$ and the graph $G_M$. A possible matching edge $e=\{u,v\}\in E$ is termed \emph{preferred by $u$} in $M$ if $u$ is either (1) unmatched in $M$ or (2) matched by edge $\{u,w\} \in M$ such that $v \succ_u w$. The edge $e \in E$ is termed a \emph{blocking pair} in $M$ if it is both preferred by $u$ and $v$. Further, we term the pair $\{u,v\}$ \emph{accessible} in $M$ if $u$ and $v$ are connected via a path of length at most 2 in $G_M$. A possible matching edge $e=\{u,v\}\in E$ is termed a \emph{local blocking pair} in $M$ if it is both a blocking pair and accessible in $M$. Thus, for a local blocking pair $e$ both agents can strictly increase their utility by deviating jointly to $e$ (and dismissing any existing incident edges). A state $M$ without a local blocking pair is termed a \emph{locally stable matching}.

Most of our analysis concerns iterative round-based dynamics. In this process, we pick in each step one local blocking pair $e = \{u,v\}$, remove all edges incident to $u$ and $v$ in $M$, and then add $e$ to $M$. We call one such step a \emph{local improvement step}. By \emph{random dynamics} we refer to the process when in each step the local blocking pair is chosen uniformly at random from the ones available. Consider a local blocking pair $\{u,v\} \in E$ that is resolved in such a step. Before the step, $u$ and $v$ are either connected directly by a link $\{u,v\} \in L$, by a path of two links $\{u,w\},\{w,v\} \in L$, or by a path of a single matching edge and a single link in $L$. In the latter case, let w.l.o.g. $\{u,w\}$ be the matching edge and $\{v,w\}$ the link. Since $\{u,w\}$ is the only edge of $M$ incident to $u$, the local improvement step will remove $\{u,w\}$ to create $\{u,v\}$. This observation will be helpful in our constructions below, and for simplicity we will refer to such a step as "an edge moving from $\{u,w\}$ to $\{u,v\}$" or "$u$'s edge moving from $w$ to $v$".

In subsequent sections, we will consider dynamics in which each agent has a memory that allows to "remember" one matching partner from a previous step. Memory can be thought of as a cache of size 1 to store a single previous partner. In this case, we extend the notion of accessible agents as follows. A pair $\{u,v\}$ of agents becomes accessible in $M$ if (1) there is a $u$-$v$-path of distance at most 2 in $G_M$, or (2) if $u$ is present in the memory of $v$, or $v$ in the memory of $u$. All other definitions follow accordingly. Hence, in this case a local blocking pair can be based solely on access through memory. We consider three strategies for memory update. For \emph{random memory}, we assume that in every step the memory contains a previous matching partner chosen uniformly at random. For \emph{recency memory}, each agent keeps in memory the last matching partner that is different from the current partner. For \emph{quality memory}, each agent remembers the previous matching partner that gave him the highest utility.

\section{Reachability in Bipartite Network Games}
\label{sec:reach}

In this section we focus on lower bounds for the \Reach\ problem in bipartite network games. Throughout, we focus on the empty matching as a natural candidate for the initial matching and show that \Reach\ is \classNP-hard to decide. This is in contrast to correlated network games, where \Reach\ is always true, and for every initial matching there is a polynomial path to a locally stable matching~\cite{Hoefer13}. However, we show that given a distinct matching to reach, deciding \Reach\ becomes \classNP-hard, even for correlated job-market games. 

Additionally, we give a network game and an initial matching such that we need an exponential number of steps before reaching any locally stable matching. This is again in contrast to the correlated case, in which we show that every reachable locally stable matching can be reached by a polynomial path.

\subsection{Complexity}
\label{sec:complex}

Our constructions rely on a simple instance, in which it is impossible to reach a locally stable matching from the empty initial matching. We term this structure a \emph{circling gadget}.

\begin{bspi}[Circling Gadget] \rm
The agent set $V$ consists of the classes $U=\{1,2,3\}$ and $W=\{A,B,C,b_1,b_2,b_3\}$. The links are $L = \{\{A,b_1\},\{B,b_2\},\{C,b_3\},$ $\{1,b_1\},$ $\{2,b_2\},$ $\{3,b_3\},$ $\{A,B\},$ $\{B,C\},$ $\{C,A\}$ (see Fig.~\ref{fig:circling} below). For simplicity, we restrict the possible matching edges to $E = \{\{u,v\}\mid u = 1,2,3, v = A,B,C\}$. It will become obvious that by ranking $b_1, b_2, b_3$ at the bottom, we can similarly allow $E = U \times W$ without changing the arguments.

\begin{figure}
\setlength{\tabcolsep}{20pt}
\begin{tabular*}{\textwidth}{cc}
\begin{tabular}[b]{|c|c|}\hline
Agent & Preference List \\\hline\hline
$1 $ & $C\succ B\succ A $ \\\hline
$2 $ & $A\succ C\succ B $ \\\hline
$3 $ & $B\succ A\succ C $ \\\hline
$A $ & $3\succ 1\succ 2 $ \\\hline
$B $ & $1\succ 2\succ 3 $ \\\hline
$C $ & $2\succ 3\succ 1 $ \\\hline
\end{tabular} &
\begin{tikzpicture}[thick,scale=0.9, every node/.style={transform shape}]
\tikzstyle{vertex}=[circle,draw,minimum size=2pt]

\node[vertex] (1) at (1,3) {$1$};  
\node[vertex] (2) at (3.5,4) {$2$};
\node[vertex] (3) at (6,3) {$3$};
\node[vertex] (A) at (1,1) {$A$};
\node[vertex] (B) at (3.5,2) {$B$};
\node[vertex] (C) at (6,1) {$C$};
\node[vertex, scale=0.5] (b1) at (1,2) {$b_1$}; 
\node[vertex, scale=0.5] (b2) at (3.5,3) {$b_2$};
\node[vertex, scale=0.5] (b3) at (6,2) {$b_3$};

\path[link] (1) -- (b1);
\path[link] (2) -- (b2);
\path[link] (3) -- (b3);
\path[link] (A) -- (b1);
\path[link] (B) -- (b2);
\path[link] (C) -- (b3);
\path[link] (A) -- (B);
\path[link] (B) -- (C);
\path[link] (C) -- (A);

\end{tikzpicture}\\
\end{tabular*}
\caption{\label{fig:circling} Preferences and links in the circling gadget.}
\end{figure}
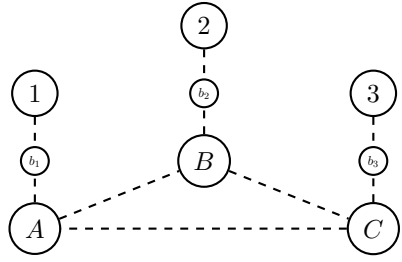

This gadget has two locally stable matchings, namely 
$$\{\{1,B\},\{2,C\},\{3,A\}\} \; \text{ and }  \; \{\{1,C\},\{2,A\},\{3,B\}\}\enspace.$$ 
However, we will show that from every state in which one agent of $W$ is unmatched, every possible sequence of local improvement steps leads to some state where some agent of $W$ is unmatched again. The symmetry of the gadget then allows to repeat the argument, and thus no locally stable matching can be reached. For our analysis, let us w.l.o.g.\ assume that $1$ is unmatched. 

First, suppose $A$ is not matched to $3$. Then $\{A,1\}$ is a local blocking pair, since the agents are at hop distance 2 from each other, $1$ is unmatched and $A$ prefers $1$ over $2$. When creating this edge, the distance from $1$ to $B$ reduces from 3 to 2. Agent $1$ is the top choice of $B$. Hence, as long as $1$ remains matched to $A$, $\{B,1\}$ remains a local blocking pair. When $A$ gets matched to some other agent, then $1$ becomes unmatched. Otherwise, if we create $\{B,1\}$, $A$ is left unmatched, which implies that also one of the agents $2$ or $3$ must be unmatched.

Now suppose $A$ is matched to $3$. Following the arguments above, as long as $B$ is matched to $2$, $\{C,2\}$ is a local blocking pair. Hence, suppose that $B$ is not matched to 2. Since $1$ is unmatched, $\{B,3\}$ is a local blocking pair. If we create this edge, $A$ becomes unmatched, and we reach the case analyzed above. Hence, every sequence of moves yields a state in which another agent of $\{1,2,3\}$ is unmatched.

In turn, it is simple to verify that when agent $A$ is matched to some more preferred partner outside of the gadget, the remaining agents can always stabilize easily through local improvement steps. For example, starting from the empty matching, consider the following sequence: $\emptyset \to \{(2,B)\} \to \{(2,C)\}$. The latter is a locally stable matching when $A$ is matched to a preferred partner outside the gadget. 

Finally, observe that all the previous arguments can be made if $E = U \times W$ and the agents $b_i$ are ranked arbitrarily at the bottom of the lists of agents in $U$. Consider agent $b_1$, the argument is similar for all other agents $b_i$. $b_1$ is only accessible with agent $i' \in U$, $i'\neq 1$, when $i'$ is matched to $A$. At this point, however, $i'$ has no incentive to deviate to $b_1$. Hence, each $b_i$ can only be involved in the local blocking pair $\{i,b_i\}$. Since these edges also exist as links, their creation does not change any of the arguments made above. \hfill $\blacksquare$
\end{bspi}

In our hardness proofs we use this gadget if we need to force a certain agent to be matched in a locally stable matching. For this we identify the forced agent with $A$ of the gadget and declare all allowed outside connections preferable to the gadget agents. Then, if the agent is matched to a partner outside the gadget, the gadget can stabilize while otherwise it does not.

Our first theorem proves \classNP-hardness of deciding \Reach\ to a given locally stable matching. In our related work, we provided a simpler proof template for this statement in a variety of matching games with additional visibility or externality constraints~\cite{HoeferW14}. The template, however, does not work for the special case of a job-market game (i.e., links exist only within one partition) with correlated preferences (i.e., preferences are correlated by a single value $b(e)$ for each matching edge). Since the empty matching is always locally stable in job market games, we use a different initial state to show the result.

In addition, the proof will introduce the main construction and intuition that will be used to show a number of hardness proofs for different scenarios with and without memory throughout the paper.

\begin{satz}
\label{satz6}
It is \classNP-hard to decide \Reach\ to a given locally stable matching in a correlated job-market game.
\end{satz}

\begin{proof}
We reduce from \ThreeSAT. Let the given \ThreeSAT-formula contain $k$ variables $x_1,\ldots,x_k$ and $l$ clauses $C_1,\ldots,C_l$, where clause $C_j$ holds the literals $l_{1j}, l_{2j}$ and $l_{3j}$. For each variable $x_i$, our instance contains an agent $u_{x_i}$ in $U$, and four agents $a_{x_i}, x_i, \overline{x}_i$ and $v_{x_i}$ in $W$. For every clause $C_j$, our instance contains an agent $u_{C_j}$ in $U$ and two agents $a_{C_j}$ and $v_{C_j}$ in $W$. Further, there is an agent $a$ in $W$. We allow all possible matching edges $E = U \times W$, but due to the structure of the instance and our choice of initial state, some of these edges can never be created during the dynamics.

First, let us describe the link structure. Since we have a job market game, links are all among agents in $W$. The network $N$ forms a path with a branching in the middle, see Fig.~\ref{fig:3satJob}. The first part of the path contains links $\{a_{C_j},a_{C_{j+1}}\}$ for $j = 1,\ldots,l-1$, a link $\{a_{C_l},a_{x_1}\}$, links $\{a_{x_i},a_{x_{i+1}}\}$ for $i=1,\ldots,k-1$, and a link $\{a_{x_k},a\}$. The branching is created by links $\{a,x_i\}$, $\{a,\bar{x}_i\}$, $\{x_i,v_{x_1}\}$ and $\{\bar{x}_i,v_{x_1}\}$ for $i=1,\ldots,k$. Finally, we obtain the second part of the path by links $\{v_{x_j},v_{x_{j+1}}\}$ for $j = 1,\ldots,l-1$, a link $\{v_{x_k},v_{C_1}\}$, and links $\{v_{C_j},v_{C_{j+1}}\}$ for $j=1,\ldots,l-1$.

The initial matching $M_0$ consists of $M_0 = \{\{u_s,a_s\} \mid  s \in\{x_1,..,x_k\}\cup\{C_1,..,C_l\}\}$. The given locally stable matching to be reached is $M= \{\{u_s,v_s\} \mid  s \in\{x_1,..,x_k\}\cup\{C_1,..,C_l\}\}$. In Table~\ref{tab:3satJob} we provide a description of the preferences w.r.t.\ the relevant matching edges in the dynamics. All other matching edges not mentioned in the table can be considered to have edge weight of, say, $\varepsilon \ll 1$. Thereby, these matches are ranked arbitrarily at the bottom of the preference list of each involved agent.
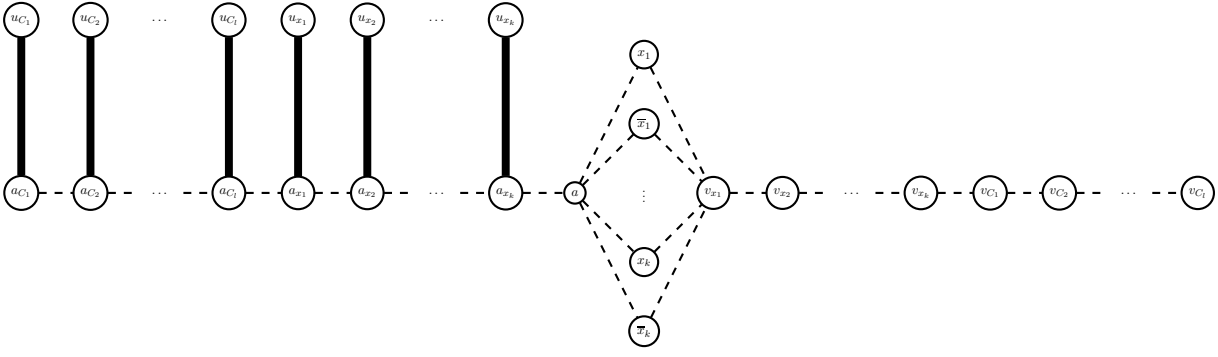
\begin{figure}
\begin{center}
\begin{tikzpicture}[thick,scale=0.46, every node/.style={transform shape}]
\tikzstyle{vertex}=[circle,draw,minimum size=5pt,scale=1]

\node[vertex] (x1) at (2,8) {$x_1$};  
\node[vertex] (-x1) at (2,6) {$\overline{x}_1$};

\node[vertex] (xk) at (2,2) {$x_k$};
\node[vertex] (-xk) at (2,0) {$\overline{x}_k$}; 

\node[vertex] (s) at (0,4) {$a$};

\node[vertex] (ex1) at (4,4) {$v_{x_1}$}; 
\node[vertex] (ex2) at (6,4) {$v_{x_2}$};
\node[vertex] (exk) at (10,4) {$v_{x_k}$};
\node[vertex] (ec1) at (12,4) {$v_{C_1}$};
\node[vertex] (ec2) at (14,4) {$v_{C_2}$}; 
\node[vertex] (ecl) at (18,4) {$v_{C_l}$};

\node[vertex] (uc1) at (-16,9) {$u_{C_1}$};
\node[vertex] (uc2) at (-14,9) {$u_{C_2}$}; 
\node[vertex] (ucl) at (-10,9) {$u_{C_l}$};
\node[vertex] (ux1) at (-8,9) {$u_{x_1}$}; 
\node[vertex] (ux2) at (-6,9) {$u_{x_2}$};
\node[vertex] (uxk) at (-2,9) {$u_{x_k}$};

\node[vertex] (ac1) at (-16,4) {$a_{C_1}$};
\node[vertex] (ac2) at (-14,4) {$a_{C_2}$}; 
\node[vertex] (acl) at (-10,4) {$a_{C_l}$};
\node[vertex] (ax1) at (-8,4) {$a_{x_1}$}; 
\node[vertex] (ax2) at (-6,4) {$a_{x_2}$};
\node[vertex] (axk) at (-2,4) {$a_{x_k}$};

\node (dots1) at (8,4) {$~~\ldots~~$}; 
\node (dots2) at (16,4) {$~~\ldots~~$};
\node (dots3) at (2,4) {$\vdots$};
\node (dots4) at (-4,9) {$~~\ldots~~$}; 
\node (dots5) at (-12,9) {$~~\ldots~~$};
\node (dots6) at (-4,4) {$~~\ldots~~$}; 
\node (dots7) at (-12,4) {$~~\ldots~~$};

\path[link] (ac1) -- (ac2) -- (dots7) -- (acl) -- (ax1) -- (ax2) -- (dots6) -- (axk) -- (s) -- (x1) -- (ex1) -- (ex2) -- (dots1) -- (exk) -- (ec1) -- (ec2) -- (dots2) -- (ecl);
\path[link] (s) -- (-x1) -- (ex1);
\path[link] (s) -- (xk) -- (ex1);
\path[link] (s) -- (-xk) -- (ex1);

\path[existEdge]
(uc1) edge (ac1)
(uc2) edge (ac2)
(ucl) edge (acl)
(ux1) edge (ax1)
(ux2) edge (ax2)
(uxk) edge (axk);
\end{tikzpicture}
\end{center}
\caption{\label{fig:3satJob} Network and initial matching in the instance for Theorem~\ref{satz6}}
\end{figure}

\begin{table}
\begin{center}
\setlength{\tabcolsep}{15pt}
\begin{tabular}{|c|c|l|l|}\hline
$u\in U$ & $w\in W$ & edge weight $b(\{u,w\})$ & \\\hline\hline
$u_{C_j}$ & $a_{C_{j'}}$           & $2(k+l+1)j + j'              $ & $j,j' \in[l]$\\\hline
$u_{C_j}$ & $a_{x_i}   $           & $2(k+l+1)j + l+i             $ & $i \in [k], j \in[l]$\\\hline
$u_{x_i}$ & $a_{C_j}   $           & $2(k+l+1)(l+i) +j            $ & $i \in [k], j \in[l]$\\\hline
$u_{x_i}$ & $a_{x_{i'}}$           & $2(k+l+1)(l+i) +l+i'         $ & $i,i' \in [k]$\\\hline
$u_{C_j}$ & $a$                    & $2(k+l+1)j + k+l+1           $ & $j \in[l]$\\\hline
$u_{x_i}$ & $a$                    & $2(k+l+1)(l+i) + k+l+1       $ & $i \in [k]$\\\hline
$u_{C_j}$ & $l_{1j}/l_{2j}/l_{3j}$ & $2(k+l+1)j + k+l+2           $ & $j \in[l]$\\\hline
$u_{x_i}$ & $x_i/\bar{x}_i$        & $2(k+l+1)(l+i) + k+l+2       $ & $i \in [k]$\\\hline
$u_{C_j}$ & $v_{x_{i'}}$           & $2(k+l+1)j + (k+l+2) + i'    $ & $i \in [k], j \in[l]$ \\\hline
$u_{x_i}$ & $v_{x_{i'}}$           & $2(k+l+1)(l+i) + (k+l+2) + i'$ & $i,i' \in [k], i' \le i$ \\\hline
$u_{C_j}$ & $v_{C_{j'}}$           & $2(k+l+1)j + (k+l+2) + k + j'$ & $j,j' \in [l], j' \le j$ \\\hline
%
\end{tabular}
\end{center}
\caption{\label{tab:3satJob} Preferences of agents in the instance for Theorem~\ref{satz6}. Since preference lists are correlated, we here specify the edge weight for a significant subset of matching edges. Agents rank their possible partners in non-increasing order of edge weight. For each edge we did not specify here, we assume there is a very small weight. Edges of small weight are then ranked in arbitrary order at the bottom the preference lists of the incident agents. They do not change the dynamics as described in the proof.}
\end{table}
Let us first provide a broad intuition of the dynamics in the gadget. Consider the display of the gadget in Fig.~\ref{fig:3satJob}. Agents in $U$ will generally prefer to be matched ``further to the right'' along the path in $N$ (i.e., most preferred $v_{C_l}$, least preferred $a_{C_1}$). We assume all agents $x_i$ and $\bar{x_i}$ in the branching are equally preferred for each agent $u \in U$, but in order to avoid ties in the preference lists we can also break them arbitrarily. Agents in $W$ generally prefer to have a partner that is ``further to the right'' in the order the agents of $U$ are displayed (i.e., most preferred $u_{x_k}$, least preferred $u_{C_1}$). As such there is, e.g., exactly one local blocking pair for the initial matching, namely $\{u_{x_k},a\}$. 

Since the instance is a job-market game, if an agent in $U$ becomes unmatched, it can never be involved in any local blocking pair and thus stays unmatched from that point on. However, since in $M$ all agents of $U$ are matched, they must stay matched throughout the sequence. Therefore, the sequence will proceed from $M_0$ by ``shifting the edges along the path'' -- for every $s \in \{x_1,\ldots,x_k,C_1,\ldots,C_l\}$, agent $u_s$ will become matched, one by one, to the agents of $W$ along the left side of the path leading from $a_s$ to $a$, one agent in the branching, and then the agents on right side leading from the branching to $v_s$. This is the only way $u_s$ will be able to discover and match to $v_s$ via resolution of local blocking pairs. Hence, the matching edge incident to $u_s$ will ``move'' along the path, through the branching and then to $v_s$.

We now describe some adjustments to the general outline above that serve to create a correct reduction. If agent $u_{x_i}$ is matched to $a$, we restrict his preference so that the only local blocking pairs involving $u_{x_i}$ are $\{u_{x_i},x_i\}$ or $\{u_{x_i},\bar{x}_i\}$. Hence, at this point $u_{x_i}$ is only interested in deviating to designated agents $x_i, \bar{x}_i$. Matching edges $\{u_{x_i}, s\}$ for $s \in \{x_j,\bar{x}_j\}$ with $i \neq j$ become ranked arbitrarily at the bottom of the preference list of the corresponding agents. A similar adjustment is done for the clause agents. If agent $u_{C_j}$ is matched to $a$, we restrict his preference so that the only local blocking pairs are $\{u_{C_j},l_{j1}\}$, $\{u_{C_j},l_{2j}\}$ and $\{u_{C_j},l_{3j}\}$. Hence, at this point $u_{C_j}$ is only interested in deviating to agents that correspond to literals that evaluate the formula to true.

Furthermore, we ``stop'' the movement of the matching edge incident to $u_{x_i}$ once it has reached to $v_{x_i}$ -- we assume edges $\{u_{x_i},v_{x_j}\}$ for $j > i$ and $\{u_{x_i},v_{C_j}\}$ for $j=1,\ldots,l$ are ranked at the bottom of the lists. A similar adjustment is made to edges $\{u_{C_j},v_{C_{j'}}\}$ for $j' > j$. Observe that these adjustments imply that $M$ is a locally stable matching. More precisely, $M$ becomes the unique (globally) stable matching.

Now it is rather simple to see the idea of the reduction. The key property is that once any edge $\{u_{x_i},v_{x_i}\}$ exists, every agent $u_{C_j}$ must be matched to $v_{x_{i'}}$ with $i' > i$ or $v_{C_{j'}}$. Assume otherwise, then the edge incident to some $u_{C_j}$ must still pass through $v_{x_i}$ to reach $v_{C_j}$. Since $v_{x_i}$ is the most preferred partner of $u_{x_i}$, the latter would have to become (and remain) unmatched, at which point $M$ could not be obtained. By similar arguments, $M$ can only be reached if and only if every agent $v_{x_i}$ becomes matched to $u_{C_l},\ldots,u_{C_1},u_{x_k},\ldots,u_{x_i}$ in that order. 

Finally, for the formal correctness proof, suppose the \ThreeSAT\ formula is satisfiable. Then for each variable $x_1,\ldots,x_k$ in that order, we match each $x_i$ iteratively along the path to $a$ and then match it to $x_i$ ($\bar{x_i}$) if $x_i$ is set false (true) in the satisfying assignment. This leaves unmatched the agents of $W$ in the branching that correspond to the satisfying assignment. Now, for each clause $C_l,\ldots,C_1$ in this order, we consider agent $u_{C_j}$, iteratively match it along the path to $a$, then to one of the literal agents $l_{1j}$, $l_{2j}$, $l_{3j}$ (one is unmatched since the assignment is satisfying), and further to $v_{C_j}$. For each variable $x_k,\ldots,x_1$ in that order, we then match $u_{x_i}$ from its current partner to $v_{x_1}$ and further to $v_{x_i}$. Thereby we obtain a sequence from $M_0$ to $M$.

In turn, suppose there is a sequence from $M_0$ to $M$. Due to the key property, we must match every $u_{x_i}$ to one of its designated agents in the branching before matching any of them to $v_{x_1}$. This induces an assignment of the variables. Then, since $u_{C_j}$ considers deviating from $a$ only to agents that correspond to its literals, each $u_{C_j}$ must be matched to one of its literal agents in order to become accessible to $v_{C_j}$. Therefore, to construct the desired sequence, the \ThreeSAT\ formula must be satisfiable.
\end{proof}

The main property of the previous proof is that, intuitively, we need to carefully match the variable agents of $U$ in a branching in order to allow the clause agents of $U$ reach their desired partner. We now drop the requirement of a job-market game and allow links among both partitions of the agents. Thereby, we can even establish the same result starting from the empty matching. This is again a stronger result than the one derived by the template in~\cite{HoeferW14}, since the template does not apply to an empty initial matching. 

\begin{satz}
\label{satz1}
It is \classNP-hard to decide \Reach\ from the initial matching $M = \emptyset$ to a given locally stable matching in a correlated bipartite network game.
\end{satz}
\begin{proof}
The argument will heavily rely on the ideas put forward in the previous proof. The main adjustment here is that we replace the left side of the path (agents $a_{x_i}$ and $a_{C_j}$, for $i\in[k], j\in[l]$) by a different construction that allows to create matching edges for unmatched agents in $U$. In contrast to the previous construction, we add agents $b_h$, $h=1,\ldots,l+k-1$, to $U$, and replace all agents $a_{x_i}$ and $a_{C_j}$ by a single agent $a_1$ in $W$. For the link set $L$, consider the display of the instance in Fig.~\ref{fig:3satEmpty}. The $b$-agents are not labeled in the picture as they only act as a buffer to ensure that local blocking pairs involving $b_h$ exist only when the neighboring $u$-agent is matched to $a$. We again construct a path with a branching in a similar fashion as above. For the left side of the path, we add links $\{a,a_1\},$ $\{a_1,u_{C_1}\},$ $\{u_{C_j},b_j\}$ for $j=1,\ldots,l$, $\{b_j,u_{C_{j+1}}\}$ for $j=1,\ldots,l-1$, $\{b_l,u_{x_1}\},$ $\{u_{x_i},b_{l+i}\}$ for $i=1,\ldots,k-1$ and $\{b_{l+i},u_{x_{i+1}}\}$ for $i=1,\ldots,k-1$. We start from the empty matching, and the the goal is again to reach $M=\{\{u_s,v_s\} \mid  s \in\{x_1,..,x_k\}\cup\{C_1,..,C_l\}\}$. For a formal description of the correlated preference lists, see Table~\ref{tab:3satEmpty}.

\begin{figure}
\begin{center}
\begin{tikzpicture}[thick,scale=0.43, every node/.style={transform shape}]
\tikzstyle{vertex}=[circle,draw,minimum size=5pt,scale=1]

\node[vertex] (x1) at (2,8) {$x_1$};  
\node[vertex] (-x1) at (2,6) {$\overline{x}_1$};

\node[vertex] (xk) at (2,2) {$x_k$};
\node[vertex] (-xk) at (2,0) {$\overline{x}_k$}; 

\node[vertex] (s) at (0,4) {$a$};
\node[vertex] (s1) at (-9,3) {$a_1$};

\node[vertex] (ex1) at (4,4) {$v_{x_1}$}; 
\node[vertex] (ex2) at (6,4) {$v_{x_2}$};
\node[vertex] (exk) at (10,4) {$v_{x_k}$};
\node[vertex] (ec1) at (12,4) {$v_{C_1}$};
\node[vertex] (ec2) at (14,4) {$v_{C_2}$}; 
\node[vertex] (ecl) at (18,4) {$v_{C_l}$};

\node[vertex] (uc1) at (-16,9) {$u_{C_1}$};
\node[vertex] (uc2) at (-14,9) {$u_{C_2}$}; 
\node[vertex] (ucl) at (-10,9) {$u_{C_l}$};
\node[vertex] (ux1) at (-8,9) {$u_{x_1}$}; 
\node[vertex] (ux2) at (-6,9) {$u_{x_2}$};
\node[vertex] (uxk) at (-2,9) {$u_{x_k}$};

\node[vertex] (bx1) at (-7,9) {};
\node[vertex] (bcl) at (-9,9) {};
\node[vertex] (bc1) at (-15,9) {};

\node (dots1) at (8,4) {$~~\ldots~~$}; 
\node (dots2) at (16,4) {$~~\ldots~~$};
\node (dots3) at (2,4) {$\vdots$};
\node (dots4) at (-4,9) {$~~\ldots~~$}; 
\node (dots5) at (-12,9) {$~~\ldots~~$};

\path[link] (s1) -- (s) -- (x1) -- (ex1) -- (ex2) -- (dots1) -- (exk) -- (ec1) -- (ec2) -- (dots2) -- (ecl);
\path[link] (s1) -- (uc1) -- (bc1) -- (uc2) -- (dots5) -- (ucl) -- (bcl) -- (ux1) -- (bx1) -- (ux2) -- (dots4) -- (uxk);
\path[link] (s) -- (-x1) -- (ex1);
\path[link] (s) -- (xk) -- (ex1);
\path[link] (s) -- (-xk) -- (ex1);

\end{tikzpicture}
\end{center}
\caption{\label{fig:3satEmpty} Network in the instance for Theorem~\ref{satz1}}
\end{figure}
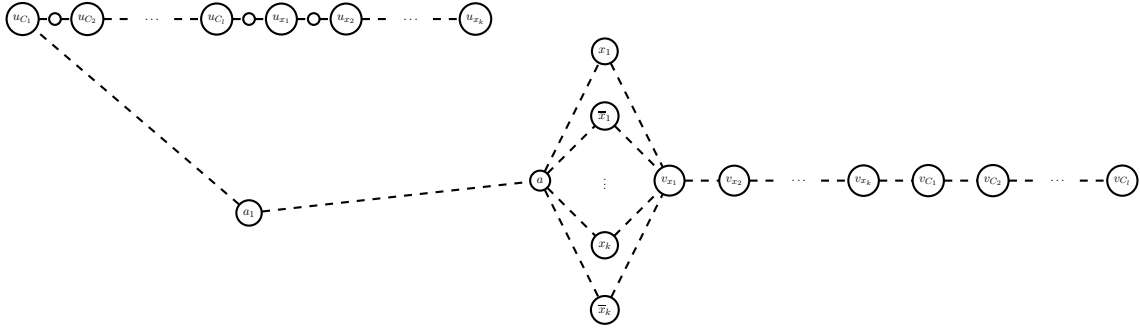

\begin{table}
\begin{center}
\setlength{\tabcolsep}{15pt}
\begin{tabular}{|l|c|l|l|}\hline
$u\in U$  & $w\in W$ & edge weight $b(\{u,w\})$ & \\\hline\hline
$u_{C_j}$ & $a$              & $j$                         & $j \in[l]$\\\hline
$u_{x_i}$ & $a$              & $l+i$                       & $i \in [k]$\\\hline
$b_h$     & $a$              & $h+\frac{1}{2}$             & $h \in [k+l-1]$\\\hline
$u_{C_j}$ & $l_{1j}/l_{2j}/l_{3j}$ & $k+l+1$                     & $j \in[l]$\\\hline
$u_{x_i}$ & $x_i/\bar{x}_i$  & $k+l+1$                     & $i \in [k]$\\\hline
$u_{C_j}$ & $v_{x_i}$        & $(k+l+1) + (k+l)j + i$      & $i \in [k], j \in[l]$\\\hline
$u_{x_i}$ & $v_{x_{i'}} $    & $(k+l+1) + (k+l)(l+i) + i'$ & $i,i' \in [k], i' \le i$\\\hline
$u_{C_j}$ & $v_{C_{j'}} $    & $(k+l+1) + (k+l)j + k + j'$ & $j,j' \in[l], j' \le j$\\\hline
\end{tabular}
\end{center}
\caption{\label{tab:3satEmpty} Preferences of agents in the instance for Theorem~\ref{satz1}. For correlated preferences, we again specify the edge weights of a significant subset of matching edges, which are used to derive the preference lists. For each edge not given here, there is a very small weight, and these edges are then ranked in arbitrary order at the bottom the preference lists of the incident agents. They do not change the dynamics as described in the proof.}
\end{table}

Suppose the \ThreeSAT\ formula is satisfiable. We describe the sequence to reach $M$, which shows the general idea of the construction. First, we match $a$ iteratively to $u_{C_1}, b_1, u_{C_2}, b_2,\ldots,b_{k+l-1}, u_{x_1}$. Then, we match $u{x_1}$ to $x_1$ $(\bar{x}_1)$ if $x_1$ is false (true) in the satisfying assignment. Afterwards, we match $a$ again iteratively to $u_{C_l}, b_1, u_{C_{l-1}}, b_2,\ldots,b_{k+l-2}, u_{x_2}$, and then $u_{x_2}$ to $x_2$ $(\bar{x}_2)$ if $x_2$ is false (true) in the satisfying assignment. In similar fashion, we proceed with all variable agents. Once $u_{x_k}$ is matched to either $x_k$ or $\bar{x}_k$, we proceed with the clause agents. We match $a$ iteratively to $u_{C_1}, b_1, u_{C_2}, b_2,\ldots,b_{l-1},u_{C_l}$, match $u_{C_l}$ to an agent corresponding to a literal that evaluates $C_l$ to true, and iteratively further down the path to $v_{C_l}$. In similar fashion, we proceed with $u_{C_{l-1}},\ldots,u_{C_1}$. Finally, we match $u_{x_i}$ to $v_{x_1}$ and down the path to $v_{x_i}$, for each $i=k,\ldots,1$ in that order.

Note that agent $a$ and the path among agents $u_s$ and $b_h$ allow to create matching edges for all agents $u_s$ and shift them iteratively to agents in the branching, for $u_s$ with $s = x_1,\ldots,x_k,C_l,\ldots,C_1$ in that order. In particular, the construction requires that edges incident to $x_1,\ldots,x_k$ have to be created before any of $u_{C_1},\ldots,u_{C_l}$ is matched to an agent in the branching. Otherwise, a key property similar to the one in the previous proof implies that the desired matching cannot be reached. This is the main intuition in the following proof of the reverse direction.

Suppose there is a sequence to reach $M$ from the empty matching. With the exception of $u_{C_l}$, none of the agents from $u_{C_j}$ or $u_{x_i}$ are at hop-distance 2 from any agent of $W$ in the empty matching. Hence, the only way for such an agent $u_s$ to become matched is via resolution of local blocking pairs that evolve when the left neighbor $b_h$ is matched (cf.\ Fig.~\ref{fig:3satEmpty}). Furthermore, with the exception of $b_1$, none of the $b_h$ agents is at hop distance 2 from any agent of $W$ in the empty matching. Hence, in a similar fashion, such an agent can only get matched via a local blocking pair that evolves when the corresponding left neighbor $u_{s'}$ is matched. The latter local blocking pairs, however, can evolve only when $u_s$ is matched to $a$, since all other agents of $W$ rank all agents $b_h$ at the bottom of their lists. As a consequence, the only way to create a matching edge for an agent $b_h$ is by matching $a$ first to $u_{C_l}$, then to $b_1$, to $u_{C_{l-1}}$, to $b_2$, etc. Note that once an agent $b_h$ is matched to $a$, there is no local blocking pair involving $b_h$, since $a$ is the most preferred partner for every $b_h$. Hence, to have the sequence proceed towards matching the $v$-agents, we must match $a$ to some agent $u_s$, at which point local blocking pairs with some $x_i$- and $\bar{x}_i$-agents in the branching appear.

Now suppose some $u_{x_i}$ is matched to $a$, and we resolve a local blocking pair with an agent in the branching. Suppose $u_{x_{i-1}}$ is unmatched at this point, and suppose all $u_{x_{i'}}$ with $i' < i-1$ are matched. If $u_{x_i}$ does not become unmatched, there will be no local blocking pair involving $a$ and $u_{x_{i-1}}$. Thereby $a$ cannot become matched to $u_{x_{i-1}}$, and the latter will remain unmatched. Otherwise, if $u_{x_i}$ becomes unmatched, this can only happen when matched to some agent $v_{x_1},\ldots,v_{x_{i'}}$, at the point where that agent becomes matched to some $u_{x_{i'}}$ with $i' < i-1$. Then, however, it is easy to see that a similar key property as in the previous proof is violated -- $u_{x_{i-1}}$ will not be able to discover $v_{x_{i-1}}$, since for every agent $v_s$ matched to $u_{x_{i'}}$ there is no local blocking pair involving $u_{x_{i-1}}$. As such, $u_{x_{i-1}}$ will get ``stuck'' being matched to some agent on a path between $a$ and $v_{x_{i'}}$.

If we apply the above argument inductively, it reveals that we must resolve local blocking pairs with agents in the branching and $u_{x_i}$ in increasing order of $i$. The same argument shows that we have to match all $u_{x_1},\ldots,u_{x_k}$ to an agent in the branching before matching any of $u_{C_1},\ldots,u_{C_l}$ to agents in the branching. Together with the key property, this implies the desired structure. We need to match the agents $u_{x_i}$ in a way to leave unmatched at least one literal agent for each clause. Hence, a sequence to reach $M$ implies a satisfying assignment for the \ThreeSAT\ formula.
\end{proof}

In addition, \classNP-hardness also holds for reaching an arbitrary locally stable matching from $M = \emptyset$ under a restriction on the link structure. The proof of the following theorem was kindly provided by an anonymous referee.

\begin{satz}
\label{satzReviewer}
It is \classNP-hard to decide \Reach\ from the initial matching $M = \emptyset$ to an arbitrary locally stable matching in a bipartite network game with links restricted to $L \subseteq (U \times W) \cup (W \times W)$. 
\end{satz}

\begin{proof}
We define a stable matching $M$ to be \emph{complete} if all men and women are matched in $M$. Then, \SMTI\ is the problem of deciding whether a complete stable matching exists in an instance of the stable marriage problem where preference lists may be incomplete and may contain ties. This problem is \classNP-complete, even if each woman's list is strictly ordered and each man's list is either strictly ordered or is a tie of length 2~\cite{ManloveIIMM02}. We will use a reduction from this restriction of \SMTI\ to show that it is \classNP-hard to decide \Reach\ from the initial matching $M = \emptyset$ to an arbitrary locally stable matching in a bipartite network game. 

Let $I$ be an instance of the stated restriction of \SMTI, where $U$ is the set of men and $W$ is the set of women. We form an instance $J$ of \Reach\ by letting the potential matching edges $E$ be obtained from the graph underlying $I$. That is, $\{u,w\}$ forms an acceptable pair in $I$ if and only if $\{u,w\} \in E$. Initially, let $L = E$, i.e., the links comprise all potential matching edges. Each woman $w \in W$ initially has the same preference list in $J$ as in $I$. We create a circling gadget $G_w$ and identify woman $w$ with vertex $A_w$. That is, woman $w$ appends the men $3_w \succ 1_w \succ 2_w$, in that order, to her preference list in $J$. Let $U' \subseteq U$ be the set of men in $I$ whose preference list is a single tie. Each man $u \in U\setminus U'$ has the same preference list in $J$ as in $I$. Now let $u \in U'$ and suppose that $u$ is indifferent between $w_1$ and $w_2$ in $I$. Remove $\{u,w_1,\}$ and $\{u,w_2\}$ from $L$. Create two new women $a_u$ and $b_u$ in $J$. Add $\{u,a_u\}$ and $\{u,b_u\}$ to $E$, and add $\{a_u,w_1\}$, $\{a_u,w_2\}$, $\{a_u,b_u\}$ and $\{u,b_u\}$ to $L$. In $J$, man $u$ ranks $w_1 \succ w_2 \succ a_u \succ b_u$ in that order.
It may be verified that $I$ has a complete stable matching if and only if $J$ has a sequence to a locally stable matching from the empty initial matching. Suppose that $I$ has a complete stable matching $M$. We will show a path to a locally stable matching $M' \in J$. Pick any man $u \in U \setminus U'$. Then we may simply add the edge $\{u,M(u)\}$ to $M'$, where $M(u)$ denotes $u$'s partner in $M$. This is possible since $\{u,M(u)\} \in E \cap L$. Now let $u \in U'$. Then we add $\{u,b_u\}$ to $M'$, since $\{u,b_u\} \in E \cap L$. Then $u$ moves to $a_u$ in $M'$, since $\{a_u,b_u\} \in L$. Next $u$ moves to $M(u)$ in $M'$, since $\{M(u), a_u\} \in L$. All women in $W$ are now matched in $M'$. Therefore we can now let all circling gadgets stabilise in $M'$, so that the final matching $M'$ is locally stable in $J$.

Conversely suppose that $J$ has a sequence to a locally stable matching $M'$ starting from $M = \emptyset$. Since the circling gadgets must stabilize, every woman in $W$ is matched in $M'$ to a man in $U$ outside of the circling gadget. Thus $M = M' \cap (U \times W)$ is a complete stable matching in $I$.
\end{proof}

The following corollary provides a slightly weaker result than Theorem~\ref{satzReviewer}. It proves hardness \emph{either} under the condition $L \subseteq (U \times W) \cup (W \times W)$ \emph{or} the condition $M = \emptyset$. We include it since the proof construction is useful below to show additional hardness results for quality memory.

\begin{fol}
\label{satzNPC}
It is \classNP-hard to decide \Reach\ to an arbitrary locally stable matching in a bipartite network game with links restricted to $L \subseteq (U \times W) \cup (W \times W)$. The same result holds for the initial matching $M = \emptyset$.
\end{fol}
\begin{proof}
We adapt the previous constructions of Theorem~\ref{satz6} and~\ref{satz1}. Instead of setting a specific matching $M$ to converge to, we make use of circling gadgets. For every variable and every clause we use one separate circling gadget. We identify agents $v_{x_i}$ and $v_{C_j}$ in our previous constructions with agent $A$ of the corresponding circling gadget. In Table~\ref{tab:3satGeneral} we show how to adjust the preference of the main structure from Theorem~\ref{satz6}, which shows the result starting from $M_0 = \emptyset$. The adjustment of the main structure of Theorem~\ref{satz1} is very similar. Note that in the latter case we get $U=\{ u_s, 1_s, 2_s, 3_s, \mid s \in \{x_1,\ldots,x_k,C_1\ldots,C_l\}\}$, and hence this shows the result for $L \subseteq (U \times W) \cup (W \times W)$.
\begin{table}
\begin{center}
\begin{tabular}{|l|l|l|}\hline
Agent & Preference List &\\\hline\hline
$u_{x_i} $ & $ v_{x_i}\succ v_{x_{i-1}}\succ \ldots\succ v_{x_1}\succ x_i\succ \overline{x}_i\succ a \succ [\ldots] $ &$i \in [k]$ \\\hline
$u_{C_j} $ & $v_{C_j}\succ \ldots \succ v_{C_1}\succ v_{x_k}\succ \ldots \succ v_{x_1}\succ l_{1j}\succ l_{2j}\succ l_{3j}\succ a 
\succ [\ldots]$ & $j \in[l]$ \\\hline
$b_h$      & $a \succ [\ldots]$ & $h \in [k+l-1]$ \\\hline
$a$        & $u_{x_k}\succ b_{l+k-1} \succ u_{x_{k-1}}\succ \ldots \succ b_1 \succ u_{C_1} \succ [\ldots]$ &\\\hline
$v_{x_i} $ & $u_{C_l}\succ \ldots \succ u_{C_1} \succ u_{x_k} \succ \ldots \succ u_{x_i}\succ 3_{C_j} \succ 1_{C_j} \succ 2_{C_j} \succ [\ldots]$ & $i \in [k]$ \\\hline
$v_{C_j} $ & $u_{C_l}\succ \ldots \succ u_{C_{j+1}}\succ u_{C_j}\succ 3_{x_i} \succ 1_{x_i} \succ 2_{x_i} \succ [\ldots]$ &$j \in [l]$ \\\hline
$x_i/\overline{x}_i $ & $u_{x_i}\succ u_{C_1}\succ u_{C_2}\succ \ldots \succ u_{C_l} \succ [\ldots]$ &$i \in [k]$ \\\hline
\end{tabular}
\end{center}
\caption{\label{tab:3satGeneral} Preference lists for agents of the main structure of Corollary~\ref{satzNPC} when adding circling gadgets to the hardness instances of Theorem~\ref{satz6}. $[\ldots]$ indicates that all other remaining agents from the other partition follow in arbitrary order. The preference list of $a_1$ can be given arbitrarily.}
\end{table}

For correctness, it suffices to reason why matching edges $\{\{u_s,v_s\}\mid s\in\{x_1,\ldots,x_k,C_1,\ldots,C_l\}\}$ are required to exist in every reachable locally stable matching. Then correctness follows using the arguments of the previous proofs. Assume there exists a reachable locally stable matching with some $v_s$ not matched to any of the $u$-agents. For this $s$ the circling gadget will not stabilize. $v_{C_l}$ must be matched to $u_{C_l}$ in order to stabilize the circling gadget for $C_l$, which leaves only $u_{C_{l-1}}$ for $v_{C_{l-1}}$ etc. Inductively, all the edges $\{\{u_s,v_s\} \mid  s\in\{x_1,\ldots,x_k,C_1,\ldots,C_l\}\}$ must be present in every reachable locally stable matching.
\end{proof}


\subsection{Length of Sequences}
\label{sec:length}
For classic two-sided stable matching without locality constraints, from every initial state a stable matching can be reached in a polynomial number of steps~\cite{RothVV90}. With the restriction to local blocking pairs, we know that locally stable matchings exist, but the circling gadget shows that there are initial states from which none of the locally stable matchings is reachable. A weaker condition one might hope for is that once a state (locally stable or not) is reachable, it can also be reached via a short path of polynomially many steps. We show in Theorem~\ref{satz3} below that this condition holds for instances with correlated preferences, even beyond bipartite games in the roommates case. For uncorrelated strict preferences, however, we provide in Corollary~\ref{satz4} a class of instances with initial states such that every path to stability takes an exponential number of steps. This raises the question as to whether \Reach\ is even in \classNP.

\begin{satz}
\label{satz3}
For every network game with correlated preferences, every (locally stable) matching $M^* \subseteq E$ and initial matching $M_0 \subseteq E$ such that $M^*$ can be reached from $M_0$ through local improvement steps, there exists a sequence of at most $O(|E|^3)$ local improvement steps leading from $M_0$ to $M^*$.
\end{satz}
\begin{proof}
Consider an arbitrary sequence between $M_0$ and $M^*$. We will show that only a polynomial part of it is necessary by omitting a possibly large number of intermediate steps. We rank all edges by their benefit (allowing multiple edges to have the same rank) such that $r(e)>r(e')$ iff $b(e)>b(e')$. Further, we set $r_{max} = \max\{r(e)\mid e\in E\}$. First, recall the structure of the path of length at most 2 for accessible pairs discussed in Section~\ref{sec:model}. As a consequence, every edge $e$ formed at any point during the sequence can be assigned to have at most one \emph{direct predecessor} $e'$ in the sequence, which was necessary for $e$ to become accessible and get formed. Since $e$ and $e'$ must share a common agent, $e'$ must get removed once $e$ is formed. Hence, we call $e$ the \emph{direct successor} of $e'$. Also, note that being a local blocking pair implies $r(e') < r(e)$. Thus, every $e \in M^*$ has at most $r_{max}$ predecessors and every $e \in M_0$ at most $r_{max}$ successors. We will say that upon formation $e$ \emph{removes} its predecessor $e'$. Note that the formation of $e$ might also lead to removal of a second matching edge, which is not (assigned as) the predecessor of $e$. The remaining proof is based on three crucial observations:
\begin{enumerate}
\item[(1)] Consider an edge $e$ that has no predecessor -- that is, either $e \in M_0$, or $e \not\in M_0$ with both agents connected via a path of at most two links. For every such $e$ there is a \emph{chain of edges} resulting from direct successor/predecessor relations. We say $e$ \emph{starts a chain}.
\item[(2)] Consider $e \not\in M_0$ that starts a chain. If the final edge of the chain is not in $M^*$ and none of the edge formations of the chain removed any non-predecessor edge, the chain does not contribute in any way to the outcome and can be omitted completely from the sequence.
\item[(3)] An edge can only be removed when forming a more preferred one. Hence, every chain is limited in length by $r_{max}$.
\end{enumerate}
If $M_0 = \emptyset$, every edge that is created in the sequence is part of a chain that starts with $e \not\in M_0$. Now, by (inductive application of) (2) we can restrict to creation of the chains with final edges in $M$. Thus, there are $|M^*|$ chains, and by (3) each of them has length at most $r_{max}$. Thus, there is a sequence to $M$ of length at most $|M^*| \cdot r_{max} \in O(|E|^2)$.

If $M_0$ is arbitrary, every edge in $M_0$ starts a chain. Consider such a chain $C$ where the final edge is not in $M$. Then, the final edge of $C$ is removed by an edge that part of exactly one other chain $C'$. By (2), this can justify the existence of at most one other chain $C'$. Note that the final edge of $C'$ must have a strictly higher rank than the final edge of $C$. Hence, after repeating this argument at most $r_{max}$ many times, we arrive at a chain with a final edge in $M$. Each of these $r_{max}$ many chains might need $r_{max}$ many steps to arrive at its final edge. Thus, each edge in $M_0$ justifies the existence of at most $r_{max}$ many chains of length $r_{max}$ each. For all other edges of $M$, which cannot be traced back to a chain starting from $M_0$, one can apply the previous arguments for the empty matching. Thus, there is a sequence to $M$ of length at most $|M_0|\cdot r_{max}\cdot r_{max}+|M^*|\cdot r_{max}\in O(|E|^3)$.
\end{proof}

\begin{satz}
\label{satzExp}
There is a bipartite network game with general preferences such that a locally stable matching can be reached by a sequence of local improvement steps from the initial matching $M = \emptyset$, but every such sequence has length $2^{\Omega(|V|)}$.
\end{satz}


\begin{proof}
We compose the instance of several intertwined gadgets. Our instance consists of $n=\Omega(|V|)$ entangled gadgets $1,\ldots,n$, each of constant size, where w.l.o.g.\ let $n$ be an even number. For clarity we list and analyze them as separately as possible, see below. For a schematic overview of the network construction, see Fig.~\ref{fig:Exponential}. 

\begin{table}
\setlength{\tabcolsep}{2pt}
\begin{tabular}{p{.5\linewidth}p{.5\linewidth}}
\textbf{Generation Gadget:} &
\multirow{8}{*}{
\begin{tabular}[c]{|l|l|}\hline
Agent & Preference List \\\hline\hline
\multirow{2}{*}{$A_0^{o}$} & $End_0^{o} \succ F_1\succ 2_1\succ E_1 \succ 1_1$ \\
&$\succ D_1 \succ D_3\succ \ldots\succ D_{n-1}$\\\hline
\multirow{2}{*}{$A_0^{e}$} & $End_0^{e} \succ F_2\succ 2_2\succ E_2 \succ 1_2$ \\
&$\succ D_2 \succ D_4\succ \ldots\succ D_n$\\\hline 
$b_0 $ & $- $ \\\hline
$End_0^{o} $ & $A_0^{o} $ \\\hline
$End_0^{e} $ & $A_0^{e} $ \\\hline
\end{tabular}}\\
Agents:&\\
$A_0^{o}$, $A_0^{e}$, $b_0^o$, $End_0^{o}$, $End_0^{e}$ &\\
Links: & \\
$\{\{A_0^{o},C_i\} \mid i \text{ odd}\}, \{\{A_0^{e},C_i\} \mid i \text{ even}\}$, &\\
$\{\{b_0,D_i\} \mid i \in[n]\}$, $\{A_0^{o},b_0\}$, $\{A_0^{e},b_0\}$, &\\
$\{F_1,End_0^{o}\}$, $\{F_2,End_0^{e}\}$, $\{A_0^{o},End_0^{o}\}$,&\\
$\{A_0^{e},End_0^{e}\}$ &\\
Matching Edges: &\\
$\{\{ A_0^{o},D_i\} \mid i \text{ odd}\}, \{\{ A_0^{e},D_i\} \mid i \text{ even}\}$, &\\
$\{A_0^{o},1_1\}$, $\{A_0^{o},E_1\},\{A_0^{o},2_1\},\{A_0^{o},F_1\},$ & \\
$\{A_0^{e},1_2\}$, $\{A_0^{e},E_2\},\{A_0^{e},2_2\},\{A_0^{e},F_2\},$ & \\
$\{A_0^{o},End_0^{o}\}$, $\{A_0^{e},End_0^{e}\}$ &\\
\end{tabular}
\end{table}

\begin{table}
\setlength{\tabcolsep}{2pt}
\begin{tabular}{p{.5\linewidth}p{.5\linewidth}}
\textbf{Rotation Gadget $i=1$:}&\multirow{15}{*}{
\begin{tabular}{|l|l|}\hline
Agent & Preference List \\\hline\hline
$A_1 $ & $F_1\succ E_1 $ \\\hline
$B_1 $ & $E_1\succ D_1 $ \\\hline
$C_1 $ & $D_1\succ F_1 $ \\\hline
\multirow{3}{*}{$D_1$}  & $End2_1 \succ B_1 \succ T^B_2$\\
& $\succ F_{2}\succ 2_{2}\succ E_{2} \succ 1_{2}\succ D_{2}$\\
& $\succ S^C_2 \succ C_1\succ A_0^{o} $ \\\hline
$E_i $ & $A_1\succ B_1\succ A_0^o $ \\\hline
\multirow{3}{*}{$F_1$}  & $End1_1\succ C_1 \succ T^C_2$\\
& $\succ F_{2}\succ 2_{2}\succ E_{2} \succ 1_{2}\succ D_{2}$\\
& $\succ S^A_2 \succ A_1\succ A_0^{o} $ \\\hline
$1_1, 2_1 $ & $A_0^{o} $ \\\hline
$End1_1 $ & $F_1 $ \\\hline
$End2_1 $ & $D_1 $ \\\hline
\end{tabular}}\\
Agents:&\\
$A_1,B_1,C_1,D_1,E_1,F_1, 1_1,2_1, End1_1, End2_1$ &\\
Links: & \\
$\{A_1,B_1\},\{D_1,E_1\},\{E_1,F_1\},\{F_1,D_1\},$&\\
$\{D_1,1_1\},\{1_1,E_1\},\{E_1,2_1\},\{2_1,F_1\},$&\\
$\{A_1,End1_1\},\{C_1,End2_1\}$ &\\
Matching Edges: &\\
$\{A_1,E_1\}, \{A_1,F_1\}, \{B_1,D_1\}, \{B_1,E_1\},$&\\$\{C_1,F_1\},\{C_1,D_1\},\{F_1,End1_1\},\{D_1,End2_1\}$ &\\
&\\
&\\
&\\
&\\
\end{tabular}
\end{table}

\begin{table}
\setlength{\tabcolsep}{1pt}
\begin{tabular}{p{.5\linewidth}p{.5\linewidth}}
\textbf{Rotation Gadget $i=2$:} &\multirow{15}{*}{
\begin{tabular}{|l|l|}\hline
Agent & Preference List \\\hline\hline
$A_2 $ & $F_2\succ E_2 $ \\\hline
$B_2 $ & $E_2\succ D_2 $ \\\hline
$C_2 $ & $D_2\succ F_2 $ \\\hline
\multirow{3}{*}{$D_2$}  & $End2_2 \succ B_2 \succ T^B_3$ \\
& $\succ F_3 \succ 2_3 \succ E_3 \succ 1_3\succ D_3$\\
& $\succ S^C_3 \succ C_2 \succ D_1 \succ F_1 \succ A_0^{e}$ \\\hline
$E_2 $ & $A_2\succ B_2 \succ D_1 \succ F_1 \succ A_0^{e}$ \\\hline
\multirow{2}{*}{$F_2$} & $End1_2 \succ C_2 \succ T^C_3$\\
& $\succ F_3\succ 2_3 \succ E_3 \succ 1_3 \succ D_3 $\\
& $\succ S^A_3 \succ A_2\succ D_1 \succ F_1 \succ A_0^{e}$ \\\hline
$1_2, 2_2 $ & $D_1 \succ F_1 \succ A_0^{e}$ \\\hline
$S^C_2, T^C_2$ & $F_1$ \\\hline
$S^A_2, T^B_2$ & $D_1$ \\\hline
$End1_2 $ & $F_2 $ \\\hline
$End2_2 $ & $D_2 $ \\\hline
\end{tabular}}\\
Agents, Links, and Matching Edges similar &\\
as in all following gadgets $i=3,\ldots,n-1$ &\\
&\\
&\\
&\\
&\\
&\\
&\\
&\\
&\\
&\\
\end{tabular}
\end{table}

\begin{table}
\setlength{\tabcolsep}{1pt}
\begin{tabular}{p{.5\linewidth}p{.5\linewidth}}
\textbf{Rotation Gadget $i=3,\ldots,n-1$:} &\multirow{11}{*}{
\begin{tabular}{|l|l|}\hline
Agent & Preference List \\\hline\hline
$A_i $ & $F_i\succ E_i $ \\\hline
$B_i $ & $E_i\succ D_i $ \\\hline
$C_i $ & $D_i\succ F_i $ \\\hline
\multirow{3}{*}{$D_i$}  & $End2_i\succ B_i \succ T^B_{i+1}$\\
& $\succ F_{i+1}\succ 2_{i+1}\succ E_{i+1} \succ 1_{i+1}\succ D_{i+1}$\\
& $\succ S^C_{i+1} \succ C_i\succ D_{i-1}\succ F_{i-1}\succ A_0^{o/e}$ \\\hline
$E_i $ & $A_i\succ B_i\succ D_{i-1}\succ F_{i-1} $ \\\hline
\multirow{2}{*}{$F_i$} & $End1_i\succ C_i \succ T^C_{i+1}$\\
& $\succ F_{i+1}\succ 2_{i+1}\succ E_{i+1} \succ 1_{i+1}\succ D_{i+1}$\\
& $\succ S^A_{i+1} \succ A_i\succ D_{i-1}\succ F_{i-1}$ \\\hline
$1_i, 2_i $ & $D_{i-1}\succ F_{i-1} $ \\\hline
$S^C_i, T^C_i$ & $F_{i-1}$ \\\hline
$S^A_i, T^B_i$ & $D_{i-1}$ \\\hline
$End1_i $ & $F_i $ \\\hline
$End2_i $ & $D_i $ \\\hline
\end{tabular}}\\
Agents:&\\
$A_i,B_i,C_i,D_i,E_i,F_i,1_i, 2_i, S^A_i, S^C_i$,&\\
$T^B_i, T^C_i, End1_i, End2_i$ &\\
Links: & \\
$\{A_i,B_i\},$ $\{D_i, E_i\},$ $\{E_i,F_i\},$ $\{F_i, D_i\},$& \\ 
$\{D_i,1_i\},$ $\{1_i,E_i\},$ $\{E_i,2_i\},$ $\{2_i,F_i\},$& \\ 
$\{A_i, End1_i\},$ $\{C_i,End2_i\},$ $\{D_i,S^A_i\}$, $\{S^A_i,A_{i-1}\},$& \\ 
$\{D_i,S^C_i\}, \{S^C_i,C_{i-1}\},$ $\{F_i,T^B_i\}$, $\{T^B_i, B_{i-1}\},$&\\
$\{F_i,T^C_i\}, \{T^C_i,C_{i-1}\}$&\\
Matching Edges: &\\
$\{A_i,E_i\},$ $\{A_i,F_i\},$ $\{B_i,D_i\},$ $\{B_i,E_i\},$& \\ $\{C_i,F_i\},$ $\{C_i,D_i\},$ $\{F_i, End1_i\},$ $\{D_i,End2_i\},$& \\ $\{D_i,D_{i-1}\},$ $\{1_i,D_{i-1}\},$ $\{E_i,D_{i-1}\},$ $\{2_i,D_{i-1}\},$& \\ $\{F_i,D_{i-1}\},$ $\{D_i,F_{i-1}\},$ $\{1_i,F_{i-1}\},$&\\ $\{E_i,F_{i-1}\},$ $\{2_i,F_{i-1}\},$ $\{F_i,F_{i-1}\}$ &\\
\end{tabular}
\end{table}

\begin{table}
\setlength{\tabcolsep}{2pt}
\begin{tabular}{p{.5\linewidth}p{.5\linewidth}}
\textbf{Rotation Gadget $i=n$:}
&\multirow{14}{*}{
\begin{tabular}{|l|l|}\hline
Agent & Preference List \\\hline\hline
$A_n $ & $F_n\succ E_n $ \\\hline
$B_n $ & $E_n\succ D_n $ \\\hline
$C_n $ & $D_n\succ F_n $ \\\hline
$D_n$  & $End2_n\succ B_n\succ C_n\succ D_{n-1}\succ F_{n-1}\succ A_0^{e} $ \\\hline
$E_n $ & $A_n\succ B_n\succ D_{n-1}\succ F_{n-1}$ \\\hline
$F_n$  & $End1_n\succ C_n\succ A_n\succ D_{n-1}\succ F_{n-1}$ \\\hline
$1_n, 2_n $ & $D_{n-1}\succ F_{n-1} $ \\\hline
$S^C_n, T^C_n$ & $F_{n-1}$ \\\hline
$S^A_n, T^B_n$ & $D_{n-1}$ \\\hline
$End1_n $ & $F_n $ \\\hline
$End2_n $ & $D_n $ \\\hline
\end{tabular}}\\
Agents: &\\
$A_n,B_n,C_n,D_n,E_n,F_n, 1_n,2_n, S^A_n, S^C_n$,&\\
$T^B_n, T^C_n, End1_n, End2_n$&\\
Links: & \\
$\{A_n,B_n\},$ $\{A_n,C_n\},$ $\{B_n,C_n\},$ $\{D_n, E_n\},$&\\ $\{E_n,F_n\},$ $\{F_n, D_n\},$  $\{D_n,1_n\},$ $\{1_n,E_n\},$&\\ $\{E_n,2_n\},$ $\{2_n,F_n\},$ $\{A_n, End1_n\},$ $\{C_n,End2_n\},$&\\ $\{D_n,A_{n-1}\},$ $\{D_n,C_{n-1}\},$ $\{F_n,B_{n-1}\},$ $\{F_n,C_{n-1}\}$ &\\
Matching Edges: &\\
$\{A_n,E_n\},$ $\{A_n,F_n\},$ $\{B_n,D_n\},$ $\{B_n,E_n\},$&\\ 
$\{C_n,F_n\},$ $\{C_n,D_n\},$ $\{F_n, End1_n\},$ $\{D_n,End2_n\},$&\\ 
$\{D_n,D_{n-1}\}$, $\{1_n,D_{n-1}\}$, $\{E_n,D_{n-1}\}$,&\\ 
$\{2_n,D_{n-1}\}$, $\{F_n,D_{n-1}\},$ $\{D_n,F_{n-1}\},$ $\{1_n,F_{n-1}\},$&\\ 
$\{E_n,F_{n-1}\},$ $\{2_n,F_{n-1}\},$ $\{F_n,F_{n-1}\}$ &\\
\end{tabular}
\end{table}

\begin{figure}
\begin{center}
\begin{tikzpicture}[thick,scale=0.43, every node/.style={transform shape}]
\tikzstyle{vertex}=[circle,draw,minimum size=5pt,scale=1]

\node[vertex] (End0o) at (0,0) {$End_0^o$};  
\node[vertex] (F1) at (4,0) {$F_1$};
\node[vertex] (21) at (7,0) {$2_1$};
\node[vertex] (E1) at (10,0) {$E_1$};
\node[vertex] (11) at (13,0) {$1_1$};
\node[vertex] (D1) at (16,0) {$D_1$};

\node (aux11) at (15,-1) {};
\node (aux21) at (11,-1) {};
\node (aux31) at (9,-1) {};
\node (aux41) at (5,-1) {};
\node (aux51) at (4.5,-1.5) {};
\node (aux61) at (15.5,-1.5) {};

\node[vertex] (End11) at (0,4) {$End1_1$};
\node[vertex] (A1) at (7,4) {$A_1$};
\node[vertex] (B1) at (10,4) {$B_1$};
\node[vertex] (C1) at (13,4) {$C_1$};
\node[vertex] (End21) at (20,4) {$End2_1$};

\node[vertex] (TC2) at (2,10) {$T^C_2$};
\node[vertex] (TB2) at (2,8) {$T^B_2$};
\node[vertex] (SA2) at (18,10) {$S^A_2$};
\node[vertex] (SC2) at (18,8) {$S^C_2$};

\node[vertex] (End0e) at (0,12) {$End_0^e$};  
\node[vertex] (F2) at (4,12) {$F_2$};
\node[vertex] (22) at (7,12) {$2_2$};
\node[vertex] (E2) at (10,12) {$E_2$};
\node[vertex] (12) at (13,12) {$1_2$};
\node[vertex] (D2) at (16,12) {$D_2$};

\node (aux12) at (15,11) {};
\node (aux22) at (11,11) {};
\node (aux32) at (9,11) {};
\node (aux42) at (5,11) {};
\node (aux52) at (4.5,10.5) {};
\node (aux62) at (15.5,10.5) {};

\node[vertex] (End12) at (0,16) {$End1_2$};
\node[vertex] (A2) at (7,16) {$A_2$};
\node[vertex] (B2) at (10,16) {$B_2$};
\node[vertex] (C2) at (13,16) {$C_2$};
\node[vertex] (End22) at (20,16) {$End2_2$};

\node[vertex] (TCn1) at (2,24) {$T^C_{n-1}$};
\node[vertex] (TBn1) at (2,22) {$T^B_{n-1}$};
\node[vertex] (SAn1) at (18,24) {$S^A_{n-1}$};
\node[vertex] (SCn1) at (18,22) {$S^C_{n-1}$};

\node[vertex] (Fn1) at (4,26) {$F_{n-1}$};
\node[vertex] (2n1) at (7,26) {$2_{n-1}$};
\node[vertex] (En1) at (10,26) {$E_{n-1}$};
\node[vertex] (1n1) at (13,26) {$1_{n-1}$};
\node[vertex] (Dn1) at (16,26) {$D_{n-1}$};

\node (aux1n1) at (15,25) {};
\node (aux2n1) at (11,25) {};
\node (aux3n1) at (9,25) {};
\node (aux4n1) at (5,25) {};
\node (aux5n1) at (4.5,24.5) {};
\node (aux6n1) at (15.5,24.5) {};

\node[vertex] (End1n1) at (0,30) {$End1_{n-1}$};
\node[vertex] (An1) at (7,30) {$A_{n-1}$};
\node[vertex] (Bn1) at (10,30) {$B_{n-1}$};
\node[vertex] (Cn1) at (13,30) {$C_{n-1}$};
\node[vertex] (End2n1) at (20,30) {$End2_{n-1}$};

\node[vertex] (TCn) at (2,36) {$T^C_n$};
\node[vertex] (TBn) at (2,34) {$T^B_n$};
\node[vertex] (SAn) at (18,36) {$S^A_n$};
\node[vertex] (SCn) at (18,34) {$S^C_n$};

\node[vertex] (Fn) at (4,38) {$F_n$};
\node[vertex] (2n) at (7,38) {$2_n$};
\node[vertex] (En) at (10,38) {$E_n$};
\node[vertex] (1n) at (13,38) {$1_n$};
\node[vertex] (Dn) at (16,38) {$D_n$};

\node (aux1n) at (15,37) {};
\node (aux2n) at (11,37) {};
\node (aux3n) at (9,37) {};
\node (aux4n) at (5,37) {};
\node (aux5n) at (4.5,36.5) {};
\node (aux6n) at (15.5,36.5) {};

\node[vertex] (End1n) at (0,42) {$End1_n$};
\node[vertex] (An) at (7,42) {$A_n$};
\node[vertex] (Bn) at (10,42) {$B_n$};
\node[vertex] (Cn) at (13,42) {$C_n$};
\node[vertex] (End2n) at (20,42) {$End2_n$};

\node (auxAn) at (8,43) {};
\node (auxCn) at (12,43) {};

\node[vertex] (b0) at (27,7) {$b_0$};
\node[vertex] (A0o) at (24,1) {$A_0^o$};
\node[vertex] (A0e) at (24,13) {$A_0^e$};

\node (auxb0n) at (26,38) {};
\node (auxb0n1) at (25,26) {};

\node (auxA0n11) at (17,28.5) {};
\node (auxA0n12) at (22,28.5) {};

\node (auxA0n1) at (17,41) {};
\node (auxA0n2) at (22,41) {};

%
%
\path[link] (End0o) -- (F1) -- (21) -- (E1) -- (11) -- (D1);
\path[link] (D1) -- (aux11) -- (aux21) -- (E1) -- (aux31) -- (aux41) -- (F1) -- (aux51) -- (aux61) -- (D1);

\path[link] (End0e) -- (F2) -- (22) -- (E2) -- (12) -- (D2);
\path[link] (D2) -- (aux12) -- (aux22) -- (E2) -- (aux32) -- (aux42) -- (F2) -- (aux52) -- (aux62) -- (D2);

\path[link] (End11) -- (A1) -- (B1) -- (TB2) -- (F2) -- (TC2) -- (C1) -- (SC2) -- (D2) -- (SA2) -- (A1);
\path[link] (End21) -- (C1);

\path[link] (End12) -- (A2) -- (B2);
\path[link] (End22) -- (C2);

\path[link] (Fn1) -- (2n1) -- (En1) -- (1n1) -- (Dn1);
\path[link] (Dn1) -- (aux1n1) -- (aux2n1) -- (En1) -- (aux3n1) -- (aux4n1) -- (Fn1) -- (aux5n1) -- (aux6n1) -- (Dn1);

\path[link] (SAn1) -- (Dn1) -- (SCn1);
\path[link] (TBn1) -- (Fn1) -- (TCn1);

\path[link] (End1n1) -- (An1) -- (Bn1) -- (TBn) -- (Fn) -- (TCn) -- (Cn1) -- (SCn) -- (Dn) -- (SAn) -- (An1);
\path[link] (End2n1) -- (Cn1);

\path[link] (Fn) -- (2n) -- (En) -- (1n) -- (Dn);
\path[link] (Dn) -- (aux1n) -- (aux2n) -- (En) -- (aux3n) -- (aux4n) -- (Fn) -- (aux5n) -- (aux6n) -- (Dn);

\path[link] (End1n) -- (An) -- (Bn) -- (Cn) -- (End2n);
\path[link] (An) -- (auxAn) -- (auxCn) -- (Cn);

\node (nowhere1) at (7,21) {};
\node (nowhere2) at (9,18) {};

\node (nowhere3) at (13,21) {};
\node (nowhere4) at (11,18) {};

\node (nowhere5) at (5,20.5) {};
\node (nowhere6) at (7,17.5) {};

\node (nowhere7) at (16,20) {};
\node (nowhere8) at (15,17.5) {};

\path[link] (TCn1) -- (nowhere1);
\path[link] (C2) -- (nowhere2);

\path[link] (SAn1) -- (nowhere3);
\path[link] (A2) -- (nowhere4);

\path[link] (TBn1) -- (nowhere5);
\path[link] (B2) -- (nowhere6);

\path[link] (SCn1) -- (nowhere7);
\path[link] (C2) -- (nowhere8);

\path[link] (D1) -- (b0) -- (D2);
\path[link] (C2) -- (A0e) -- (b0) -- (A0o) -- (C1);

\path[link] (Dn) -- (auxb0n) -- (b0) -- (auxb0n1) -- (Dn1);
\path[link] (Cn) -- (auxA0n1) -- (auxA0n2) -- (A0e);

\path[link] (Cn1) -- (auxA0n11) -- (auxA0n12) -- (A0o);

\path[existEdge] (D1) -- (C1);
\path[existEdge] (F1) -- (A1);
\path[existEdge] (D2) -- (C2);
\path[existEdge] (F2) -- (A2);
\path[existEdge] (Dn1) -- (Cn1);
\path[existEdge] (Fn1) -- (An1);
\path[existEdge] (Dn) -- (Cn);
\path[existEdge] (Fn) -- (An);

\end{tikzpicture}
\end{center}
\caption{\label{fig:Exponential} Network in the instance for Theorem~\ref{satzExp}. Matching edges depict a state in the described sequence. At this state, we have generated two edges in every rotation gadget, and the exponential rotation in the gadgets starts in order to iteratively get $D_1$, $E_1$ and $F_1$ unmatched. This is necessary to match them to $A_0^o$, which thereby discovers $End_0^o$. Then, a similar approach allows $A_0^e$ to get matched to $End_0^e$.}
\end{figure}
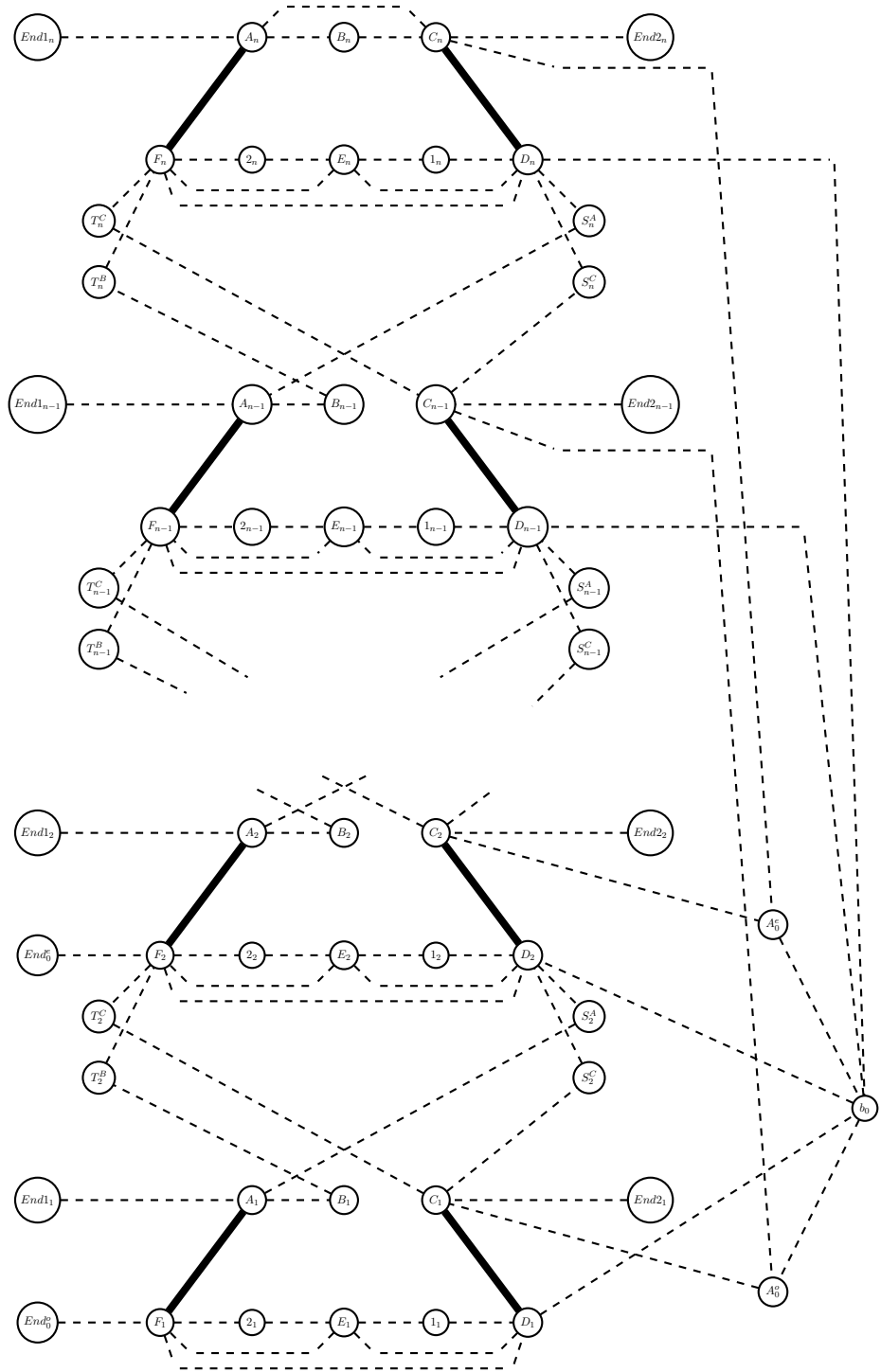

The bipartite partition of the agent set is as follows:
\begin{eqnarray*} 
U &=& \{A_0^o,End_0^e\}\cup\{A_i,B_i,C_i,End1_i, End2_i \mid i\mbox{ odd}\} \cup \{S^A_i,S^C_i,T^B_i,T^C_i \mid i \mbox{ even}\}\\
&& \hspace{0.5cm} \cup\{D_i,1_i, E_i,2_i,F_i\mid i\mbox{ even}\} \enspace,\\
W &=& \{A_0^e, b_0, End_0^o\}\cup\{A_i,B_i,C_i,End1_i, End2_i\mid i\mbox{ even}\} \cup \{S^A_i,S^C_i,T^B_i,T^C_i\mid i\mbox{ odd}\}\\
&& \hspace{0.5cm} \cup \{D_i,1_i,E_i,2_i,F_i\mid i\mbox{ odd}\}\enspace.
\end{eqnarray*}
Further, every $End$-agent is the agent $A$ of a circling gadget and ranks the agents inside the circling gadget lowest.

We first explain how to construct a particular sequence from $M_0=\emptyset$ to a locally stable matching. Then we argue that this sequence cannot be significantly shortened.

Obviously, since all agents $End1_i, End2_i$ are embedded into a circling gadget, we have to match them to $D_i,F_i$ in order to achieve a locally stable matching. Hence, every locally stable matching must contain the edges $\{D_i,End2_i\}$ and $\{F_i, End1_i\}$. When these edges exist, we say rotation gadget $i$ is \emph{fixed}. For the same reason, edges $\{A_0^o,End_0^o\}$ and $\{A_0^e, End_0^e\}$ must exist in every locally stable matching. When these edges exist, we say the generation gadget is \emph{fixed}.

The only local blocking pairs in $M=\emptyset$ are $\{A_0^o,D_i\}$ and $\{A_0^e,D_i\}$ for odd and even $i$, respectively. These agent pairs are accessible by a path of two links incident to $b_0$. We start the sequence with creation of $\{A_0^o,D_1\}$. With the link $\{A_0^o,C_1\}$, the edge $\{C_1,D_1\}$ becomes accessible. At this position we could match $D_1$ to $End2_1$, thereby stabilizing the incident circling gadget. This, however, would match $D_1$ to its most preferred partner, and thereby we would not be able to construct $\{F_1, End1_1\}$. Hence, we instead match $D_1$ iteratively to agents of rotation gadget 2. We start by moving $D_1$'s edge from $C_1$ to $S^C_2$, and then further to $D_2$, $E_2$, $F_2$, and then to $T^B_2$ and $B_1$. Now we replace $\{D_1,B_1\}$ by $\{B_1,E_1\}$, which is replaced by $\{A_1,E_1\}$ and finally by $\{A_1,F_1\}$. Then, we recreate $\{A_0^{o},D_1\}$ and $\{D_1,C_1\}$. At this point we could, in principle, fix rotation gadget 1. Then, however, we will not be able to create $\{A_0^{o},End_0^{o}\}$, since $A_0^{o}$ needs to discover $End_0^{o}$ by matching iteratively to $D_1$, $E_1$ and $F_1$. If rotation gadget 1 is fixed, $D_1$ and $F_1$ do not want to deviate to $A_0^{o}$.

Instead of fixing gadget 1, we proceed similarly with rotation gadget 2 by first creating $\{A_0^{e}, D_2\}$, then $\{C_2,D_2\}$ etc. until both $\{A_2,F_2\}$ and $\{C_2,D_2\}$ exist. Again, we will refrain from fixing rotation gadget 2, since this would imply that $A_0^{e}$ cannot discover $End_0^e$, similar to the arguments before. 

In a similar fashion we proceed with rotation gadgets $i=3,4,\ldots,n$, using $A_0^o$ or $A_0^e$ when $i$ is odd or even, resp. Note that whenever we create edges in rotation gadget $i$, rotation gadget $i+1$ is still empty and all agents $S^C_{i+1},D_{i+1},E_{i+1},F_{i+1},T^B_{i+1}$ are unmatched when $D_i$'s edge reaches them. At the end of this process, matching edges $\{C_i,D_i\}$ and $\{A_i,F_i\}$ exist for all $i=1,\ldots,n$. Observe that it is necessary to create two matching edges in each rotation gadget for the gadget to become fixed eventually. Since the only agents at which additional matching edges can be introduced are $A_0^o$ and $A_0^e$, this has to be done before the generation gadget is fixed. This shows that every $\{C_i,D_i\}$ needs to be created before the generation gadget is fixed. We do not necessarily have to advance the sequence such that the first edge in each gadget reaches position $\{A_i,F_i\}$, but we have done so in our sequence for clarity and consistency.
 
Suppose a gadget is not fixed, then the two edges in this gadget can rotate along six positions. The \emph{rotation} of a single edge runs as follows:
$$\{C_i,D_i\} \to \{B_i,D_i\} \to \{B_i,E_i\} \to \{A_i,E_i\} \to \{A_i,F_i\} \to \{C_i,F_i\} \to \{C_i,D_i\}$$
Note that since there are two edges, we always have four of the six agents matched. There is exactly one pair of agents unmatched, which allows to shift one of the edges further along its rotation. We call a sequence a \emph{rotation of gadget $i$} if the two edges in gadget $i$ both fulfill the complete cycle.

Only in gadget $n$, each of the six transitions is available directly via resolving a single local blocking pair. In all other gadgets $i=1,\ldots,n-1$, two transitions are not directly available as local blocking pairs. As we have seen above, for the transition from $\{C_i,D_i\}$ to $\{B_i,D_i\}$, $D_i$ needs to discover $B_i$ by matching iteratively to agents from gadget $i+1$. Exactly the same construction is in place for the transition $\{A_i,F_i\}$ to $\{C_i,F_i\}$, where the agents $S^A_{i+1},D_{i+1},E_{i+1},F_{i+1},T^C_{i+1}$ have to be iteratively matched to $F_i$ for this agent to discover $C_i$. Since always only one of $D_{i+1}$, $E_{i+1}$, $F_{i+1}$ is unmatched, gadget $i+1$ also needs to complete a full rotation, thereby making $D_{i+1}$, $E_{i+1}$, $F_{i+1}$ iteratively available for the less preferred partners $D_i$ or $F_i$. Consider the transition of, say, $D_i$ from $C_i$ to $B_i$. If $D_i$ is matched to $D_{i+1}$, then in order to get $E_{i+1}$ unmatched, we have to rotate the two edges in gadget $i+1$ to $D_{i+1}$ and $F_{i+1}$. This, however, would imply we ``lose'' the edge incident to $D_i$, since $D_i$ is currently still matched to $D_{i+1}$. Then, subsequently gadget $i$ cannot be fixed. Hence, in order to keep $D_i$ matched, we match it to $1_{i+1}$, proceed with the rotation in gadget $i+1$ to free $E_{i+1}$, then match $D_i$ to $E_{i+1}$ and further to $2_{i+1}$. Finally, we complete the rotation of gadget $i+1$ by freeing $F_{i+1}$, which allows $D_i$ to match to $F_{i+1}$, $T^B_{i+1}$ and $B_i$. See Fig.~\ref{fig:Rotation} for a visualization of a complete rotation in gadget $n$ to allow $D_{n-1}$ become matched to $D_n,E_n,F_n$.

\begin{figure}
\begin{center}
\begin{tikzpicture}[thick,scale=0.43, every node/.style={transform shape}]
\clip (2, 7) rectangle (17, -3);
\tikzstyle{vertex}=[circle,draw,minimum size=5pt,scale=1]

\node[vertex] (Fn) at (4,0) {$F_n$};
\node[vertex] (2n) at (7,0) {$2_n$};
\node[vertex] (En) at (10,0) {$E_n$};
\node[vertex] (1n) at (13,0) {$1_n$};
\node[vertex] (Dn) at (16,0) {$D_n$};

\node (aux1n) at (15,-1) {};
\node (aux2n) at (11,-1) {};
\node (aux3n) at (9,-1) {};
\node (aux4n) at (5,-1) {};
\node (aux5n) at (4.5,-1.5) {};
\node (aux6n) at (15.5,-1.5) {};

\node[vertex] (An) at (7,4) {$A_n$};
\node[vertex] (Bn) at (10,4) {$B_n$};
\node[vertex] (Cn) at (13,4) {$C_n$};

\node (auxAn) at (8,5) {};
\node (auxCn) at (12,5) {};

\path[link] (Fn) -- (2n) -- (En) -- (1n) -- (Dn);
\path[link] (Dn) -- (aux1n) -- (aux2n) -- (En) -- (aux3n) -- (aux4n) -- (Fn) -- (aux5n) -- (aux6n) -- (Dn);

\path[link] (An) -- (Bn) -- (Cn);
\path[link] (An) -- (auxAn) -- (auxCn) -- (Cn);

\path[existEdge] (Dn) -- (Cn);
\path[existEdge] (Fn) -- (An);
\end{tikzpicture}
%
\begin{tikzpicture}[thick,scale=0.43, every node/.style={transform shape}]
\clip (2, 7) rectangle (17, -3);
\tikzstyle{vertex}=[circle,draw,minimum size=5pt,scale=1]

\node[vertex] (Fn) at (4,0) {$F_n$};
\node[vertex] (2n) at (7,0) {$2_n$};
\node[vertex] (En) at (10,0) {$E_n$};
\node[vertex] (1n) at (13,0) {$1_n$};
\node[vertex] (Dn) at (16,0) {$D_n$};

\node (aux1n) at (15,-1) {};
\node (aux2n) at (11,-1) {};
\node (aux3n) at (9,-1) {};
\node (aux4n) at (5,-1) {};
\node (aux5n) at (4.5,-1.5) {};
\node (aux6n) at (15.5,-1.5) {};

\node[vertex] (An) at (7,4) {$A_n$};
\node[vertex] (Bn) at (10,4) {$B_n$};
\node[vertex] (Cn) at (13,4) {$C_n$};

\node (auxAn) at (8,5) {};
\node (auxCn) at (12,5) {};

\path[link] (Fn) -- (2n) -- (En) -- (1n) -- (Dn);
\path[link] (Dn) -- (aux1n) -- (aux2n) -- (En) -- (aux3n) -- (aux4n) -- (Fn) -- (aux5n) -- (aux6n) -- (Dn);

\path[link] (An) -- (Bn) -- (Cn);
\path[link] (An) -- (auxAn) -- (auxCn) -- (Cn);

\path[existEdge] (En) -- (Bn);
\path[existEdge] (Fn) -- (An);
\end{tikzpicture}
%
\begin{tikzpicture}[thick,scale=0.43, every node/.style={transform shape}]
\clip (2, 7) rectangle (17, -3);
\tikzstyle{vertex}=[circle,draw,minimum size=5pt,scale=1]

\node[vertex] (Fn) at (4,0) {$F_n$};
\node[vertex] (2n) at (7,0) {$2_n$};
\node[vertex] (En) at (10,0) {$E_n$};
\node[vertex] (1n) at (13,0) {$1_n$};
\node[vertex] (Dn) at (16,0) {$D_n$};

\node (aux1n) at (15,-1) {};
\node (aux2n) at (11,-1) {};
\node (aux3n) at (9,-1) {};
\node (aux4n) at (5,-1) {};
\node (aux5n) at (4.5,-1.5) {};
\node (aux6n) at (15.5,-1.5) {};

\node[vertex] (An) at (7,4) {$A_n$};
\node[vertex] (Bn) at (10,4) {$B_n$};
\node[vertex] (Cn) at (13,4) {$C_n$};

\node (auxAn) at (8,5) {};
\node (auxCn) at (12,5) {};

\path[link] (Fn) -- (2n) -- (En) -- (1n) -- (Dn);
\path[link] (Dn) -- (aux1n) -- (aux2n) -- (En) -- (aux3n) -- (aux4n) -- (Fn) -- (aux5n) -- (aux6n) -- (Dn);

\path[link] (An) -- (Bn) -- (Cn);
\path[link] (An) -- (auxAn) -- (auxCn) -- (Cn);

\path[existEdge] (En) -- (Bn);
\path[existEdge] (Dn) -- (Cn);
\end{tikzpicture}
%
\begin{tikzpicture}[thick,scale=0.43, every node/.style={transform shape}]
\clip (2, 7) rectangle (17, -3);
\tikzstyle{vertex}=[circle,draw,minimum size=5pt,scale=1]

\node[vertex] (Fn) at (4,0) {$F_n$};
\node[vertex] (2n) at (7,0) {$2_n$};
\node[vertex] (En) at (10,0) {$E_n$};
\node[vertex] (1n) at (13,0) {$1_n$};
\node[vertex] (Dn) at (16,0) {$D_n$};

\node (aux1n) at (15,-1) {};
\node (aux2n) at (11,-1) {};
\node (aux3n) at (9,-1) {};
\node (aux4n) at (5,-1) {};
\node (aux5n) at (4.5,-1.5) {};
\node (aux6n) at (15.5,-1.5) {};

\node[vertex] (An) at (7,4) {$A_n$};
\node[vertex] (Bn) at (10,4) {$B_n$};
\node[vertex] (Cn) at (13,4) {$C_n$};

\node (auxAn) at (8,5) {};
\node (auxCn) at (12,5) {};

\path[link] (Fn) -- (2n) -- (En) -- (1n) -- (Dn);
\path[link] (Dn) -- (aux1n) -- (aux2n) -- (En) -- (aux3n) -- (aux4n) -- (Fn) -- (aux5n) -- (aux6n) -- (Dn);

\path[link] (An) -- (Bn) -- (Cn);
\path[link] (An) -- (auxAn) -- (auxCn) -- (Cn);

\path[existEdge] (Cn) -- (Dn);
\path[existEdge] (Fn) -- (An);
\end{tikzpicture}
%
\begin{tikzpicture}[thick,scale=0.43, every node/.style={transform shape}]
\clip (2, 7) rectangle (17, -3);
\tikzstyle{vertex}=[circle,draw,minimum size=5pt,scale=1]

\node[vertex] (Fn) at (4,0) {$F_n$};
\node[vertex] (2n) at (7,0) {$2_n$};
\node[vertex] (En) at (10,0) {$E_n$};
\node[vertex] (1n) at (13,0) {$1_n$};
\node[vertex] (Dn) at (16,0) {$D_n$};

\node (aux1n) at (15,-1) {};
\node (aux2n) at (11,-1) {};
\node (aux3n) at (9,-1) {};
\node (aux4n) at (5,-1) {};
\node (aux5n) at (4.5,-1.5) {};
\node (aux6n) at (15.5,-1.5) {};

\node[vertex] (An) at (7,4) {$A_n$};
\node[vertex] (Bn) at (10,4) {$B_n$};
\node[vertex] (Cn) at (13,4) {$C_n$};

\node (auxAn) at (8,5) {};
\node (auxCn) at (12,5) {};

\path[link] (Fn) -- (2n) -- (En) -- (1n) -- (Dn);
\path[link] (Dn) -- (aux1n) -- (aux2n) -- (En) -- (aux3n) -- (aux4n) -- (Fn) -- (aux5n) -- (aux6n) -- (Dn);

\path[link] (An) -- (Bn) -- (Cn);
\path[link] (An) -- (auxAn) -- (auxCn) -- (Cn);

\path[existEdge] (En) -- (Bn);
\path[existEdge] (Fn) -- (An);
\end{tikzpicture}
%
\begin{tikzpicture}[thick,scale=0.43, every node/.style={transform shape}]
\clip (2, 7) rectangle (17, -3);
\tikzstyle{vertex}=[circle,draw,minimum size=5pt,scale=1]

\node[vertex] (Fn) at (4,0) {$F_n$};
\node[vertex] (2n) at (7,0) {$2_n$};
\node[vertex] (En) at (10,0) {$E_n$};
\node[vertex] (1n) at (13,0) {$1_n$};
\node[vertex] (Dn) at (16,0) {$D_n$};

\node (aux1n) at (15,-1) {};
\node (aux2n) at (11,-1) {};
\node (aux3n) at (9,-1) {};
\node (aux4n) at (5,-1) {};
\node (aux5n) at (4.5,-1.5) {};
\node (aux6n) at (15.5,-1.5) {};

\node[vertex] (An) at (7,4) {$A_n$};
\node[vertex] (Bn) at (10,4) {$B_n$};
\node[vertex] (Cn) at (13,4) {$C_n$};

\node (auxAn) at (8,5) {};
\node (auxCn) at (12,5) {};

\path[link] (Fn) -- (2n) -- (En) -- (1n) -- (Dn);
\path[link] (Dn) -- (aux1n) -- (aux2n) -- (En) -- (aux3n) -- (aux4n) -- (Fn) -- (aux5n) -- (aux6n) -- (Dn);

\path[link] (An) -- (Bn) -- (Cn);
\path[link] (An) -- (auxAn) -- (auxCn) -- (Cn);

\path[existEdge] (En) -- (Bn);
\path[existEdge] (Cn) -- (Dn);
\end{tikzpicture}
\end{center}
\caption{\label{fig:Rotation} A rotation of gadget $n$. Suppose $D_{n-1}$ must become iteratively matched to $D_n, E_n, F_n$ to discover $T^B_{n-1}$ and $B_{n-1}$. Top left: Starting position, $D_n$ is occupied and should be freed. Top right: After resolution of $\{B_n,D_n\}$ and $\{B_n,E_n\}$, $D_n$ is free to be matched to $D_{n-1}$. Middle left: By resolution of $\{C_n,F_n\}$ and $\{C_n,D_n\}$, $D_n$ becomes matched, so $D_{n-1}$ should become matched to $1_n$ before this step. Middle right: After resolution of $\{A_n,E_n\}$ and $\{A_n,F_n\}$, $E_n$ is free to be matched to $D_{n-1}$. Bottom left: By resolution of $\{B_n,D_n\}$ and $\{B_n,E_n\}$, $E_n$ becomes matched, so $D_{n-1}$ should get matched to $2_n$ before this step. Bottom right: After resolution of $\{C_n,F_n\}$ and $\{C_n,D_n\}$, $F_n$ is free to be matched to $D_{n-1}$, which subsequently gets matched to $T^B_{n-1}$. By resolution of $\{A_n,E_n\}$ and $\{A_n,F_n\}$, the rotation is complete.}
\end{figure}
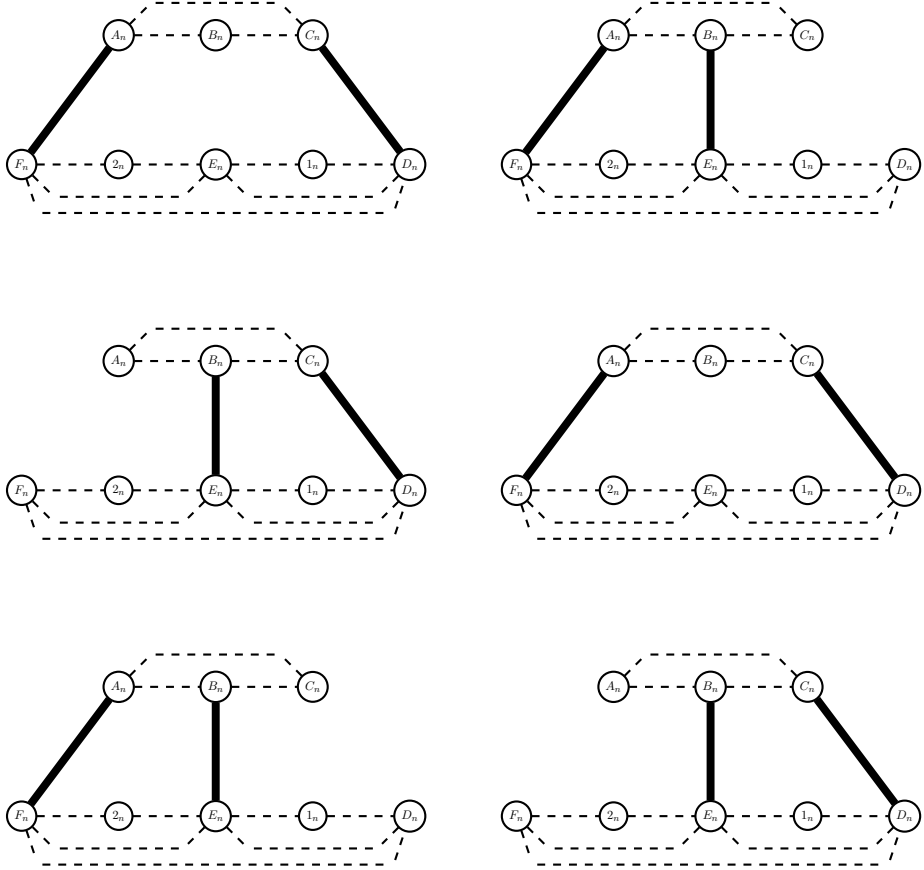

Since one rotation of gadget $i$ involves two transitions, which each require a full rotation in gadget $i+1$, a single rotation of gadget 1 requires $2^{n-1}$ rotations of gadget $n$. Now, the single rotation of gadget 1 is required, since we need to make agents $D_1, E_1, F_1$ iteratively available for matching to $A_0^{o}$ in order to fix the generation gadget. In the end, we rotate all edges back to $\{C_i,D_i\}$ and $\{A_i,F_i\}$ and fix the rotation gadgets. In this way, there is an exponential sequence in the gadget from $M = \emptyset$ to a locally stable matching.

Let us now argue, why the above idea of intertwined rotation cannot be avoided. The main underlying reason is that $A_0^o$ and $A_0^e$ must be used for generation of additional matching edges. Only after a sufficient number of matching edges is present in each gadget, we can proceed to fix the generation gadget. As part of this, we must fulfill a rotation in gadgets 1 and 2, implying an exponential number of rotations in higher gadgets. 

Obviously, we cannot fix the generation gadget before two matching edges in each rotation gadget have been created. Similarly, none of the rotation gadgets can be fixed before the generation gadget is fixed -- suppose a gadget $i$ is fixed before the generation gadget is fixed. At this point, neither $D_i$ nor $F_i$ are willing to match to $D_{i-1}$ or $F_{i-1}$. Thus, the rotation in gadget $i-1$ breaks and the gadget stabilizes. The same is true for all rotation gadgets $i-2,\ldots,1$. Hence, the rotation of gadget 1 cannot be completed, so $A_0^{o}$ is unable to reach $End_0^o$. For the same reason, we must not create three or more matching edges in a rotation gadget. This stabilizes the gadget and does not allow rotation, which is needed to eventually fix the generation gadget.

A more subtle point is the interplay of even and odd gadgets. Since we strive for bipartite instances, we cannot use a single agent $A_0$ to generate matching edges in all rotation gadgets, since there are $D_i$-agents in both partitions. Given the sequence above, the reader might wonder why $A_0^e$ has to trigger a final rotation in gadget $2$. Imagine that $A_0^e$ can directly match to $End_0^e$ once all matching edges in the even rotation gadgets are created. Then, however, we could break the intertwining of rotation gadgets: First, create all edges in odd gadgets, and move $A_0^o$ to $End_0^o$. Since at this point gadget 2 is empty, the single rotation in gadget 1 does not propagate to other gadgets. Then, create the matching edges of the even gadgets for gadget $i=2,4,\ldots,n$. This causes rotation in the odd gadget $i+1$, but since gadget $i+2$ does not yet contain any matching edges, the rotations do not propagate to higher gadgets. In this way, we could fix all gadgets in a polynomial number of steps. 

In principle, we can also adopt this approach by first creating the matching edges in odd gadgets and then in even gadgets. However, in our instance $A_0^{e}$ then triggers a final rotation in gadget 2 after generation of all matching edges in the even gadgets. At this point, all the gadgets are intertwined, and the single rotation in gadget 2 then requires that every gadget $i \ge 2$ rotates $2^{i-2}$ times.

\end{proof}


\section{Memory}
\label{sec:memory}

In this section, we focus on the impact of memory for the reachability of locally stable matchings. As a direct initial result, we observe that no memory can help with the reachability of a \emph{given} locally stable matching, even in a correlated job-market game.
\begin{fol}
\label{satz4}
It is \classNP-hard to decide \Reach\ to a given locally stable matching in a correlated job-market game with any kind of memory.
\end{fol}
\begin{proof}
We observe that the same reduction as in Theorem~\ref{satz6} yields the result even for arbitrary memory. If the \ThreeSAT\ formula is satisfiable, the sequence given in the proof can obviously be constructed without memory. Now assume that the desired matching $M^*$ is reachable. Recall the instance, for which links and initial matching are depicted in Fig.~\ref{fig:3satJob}, and the main idea of the proof. The agents of $U$ iteratively get matched along the path of agents $a_s$, into the branching of agents $x_i$ and $\bar{x}_i$, and subsequently to agents $v_s$ on the right side of the path. For every agent $u \in U$, the preference of the partner is strictly increasing during this process. Hence, for $u$ there is no (local) blocking pair with any agent $w \in W$ that $u$ was matched to before. Thus, memory can only have an impact if an agent $u$ gets unmatched. 

An agent $u \in U$ might indeed become unmatched. For example, suppose we follow the sequence outlined in the proof of Theorem~\ref{satz6} by matching all agents $u_{x_i}$ to an agent $x_i$ or $\bar{x}_i$ in the branching. Then, $u_{C_1}$ gets matched iteratively to the agents along the path, to $a$, to an agent corresponding to a true literal, and then to $v_{C_1}$. At this point, we could resolve the local blocking pair of, say, $u_{x_1}$ and $v_{x_1}$, thereby leaving $u_{C_1}$ unmatched. $u_{C_1}$ remembers all previous partners, so it forms blocking pairs, in particular, with all currently unmatched previous partners. In this way, $u_{C_1}$ can get matched again, which would be impossible without memory. However, it is clear to see in this example that once $v_{x_1}$ and $u_{x_1}$ are matched, for all $j=1,\ldots,m$, none of the edges $\{u_{C_j},v_{C_j}\}$ can be constructed -- none of them existed before, so $v_{C_j}$ would have to be discovered by $u_{C_j}$ via sequences of local blocking pair resolutions starting from $v_{x_1}$. However, $\{u_{x_1},v_{x_1}\}$ is the top choice for both agents, and therefore none of them will become part of any local blocking pair again.

It is easy to verify that this condition must hold more generally. Every agent $v_{x_j}$ must become matched to all agents $u_{C_1},\ldots,u_{C_l}$ before it gets matched to any of the agents $v_{x_k},\ldots,v_{x_j}$. The argument is quite similar to the one in Theorem~\ref{satz1}. Consider an agent $v_{x_i}$ that becomes matched to $u_{x_{i'}}$ for $i \le i'$ before getting matched to some $u_{C_j}$. Then $u_{C_j}$ does not hold $v_{x_s}$ in its memory, for all $s \ge i'$, since these agents must be discovered by iteratively resolving local blocking pairs of $u_{C_j}$ and $v_{x_1},\ldots,v_{x_{i'-1}}$. First, suppose $u_{x_{i'}}$ subsequently does not get unmatched. Then $u_{C_j}$ will not be able to discover $v_{x_s}$ for any $s \ge i'$, since no agent $v$ is involved in a local blocking pair with $u_{C_j}$ when being matched to $u_{x_{i'}}$. Also, none of these $v_{x_s}$ can be made accessible for $u_{C_j}$ by memory. Thus, figuratively speaking, the edge incident to $u_{C_j}$ gets ``stuck behind the edge incident to $u_{x_{i'}}$ when it tries to traverse the path of $v$-agents from $v_{x_1}$ to $v_{C_j}$''. 

Second, suppose $u_{x_{i'}}$ subsequently gets unmatched. This can only happen when some edge $\{u_{x_{i'}}, v_{x_s}\}$ with $i \le s \le i'-1$ exists, and is replaced by $\{u_{x_{i''}}, v_{x_s}\}$ for some $s \le i'' < i'$, since these are the only agents of $U$ that $v_{x_s}$ prefers to $u_{x_{i''}}$. Then, however, we can repeat the above argument that $u_{C_j}$ is unable to reach $v_{C_j}$, where $u_{x_{i''}}$ takes the role of $u_{x_{i'}}$ and $v_{x_s}$ takes the role of $v_{x_i}$.

This implies, in particular, that $v_{x_1}$ must be matched to all agents $u_{C_j}$ before getting matched to any of the agents $u_{x_i}$. Note that agent $u_{x_i}$ only considers deviating to agents $x_i$ and $\bar{x}_i$ from the branching. This directly implies that while being matched to the $a$-agents or the $x/\bar{x}$-agents, no agent $u_{x_i}$ can become unmatched. Therefore, we again must obey the structure of the sequence outlined in Theorems~\ref{satz6},~\ref{satz1}  before, and a satisfying assignment must exist.
\end{proof}
We will now concentrate on the impact of memory on reaching an arbitrary locally stable matching. Although existence is guaranteed for bipartite instances, even quite simple structures like the circling gadget do not allow a locally stable matching to be reached through improvement dynamics from every initial matching.

\paragraph{Quality Memory}
We start our treatment with quality memory, where every agent remembers at every round the best matching partner he ever had before. While this seems quite a natural choice and appears like a smart strategy for each agent, it can be easily fooled by starting with a much-liked partner who soon after matches with someone more preferred and never becomes unmatched again. This way the memory becomes useless which leaves us with the same dynamics as before.

\begin{figure}
\begin{center}
\begin{tikzpicture}[thick,scale=0.4]
\node[scale=0.7] (w1) at (5,5)[draw=black, circle, fill=none]{$~x~$};
\node[scale=0.7] (b1) at (5,2.5)[draw=black, circle, fill=none]{};
\node[scale=0.7] () at (5.4,6.2)[draw=none, fill=none]{$t_x>..$};
\node[scale=0.7] (w2) at (0,5)[draw=black, circle, fill=none]{$s_x$};
\node[scale=0.7] () at (-0.7,4)[draw=none, fill=none]{$t_x>..$};
\node[scale=0.7] (m1) at (2.5,2)[draw=black, circle, fill=none]{$t_x$};
\node[scale=0.7] () at (3.2,1)[draw=none, fill=none]{$s_x>x$};
\node[scale=0.7] (a1) at (8,6)[draw=none, fill=none]{};
\node[scale=0.7] (a2) at (8,4)[draw=none, fill=none]{};

\node[scale=0.7] (w11) at (17,5)[draw=black, circle, fill=none]{$~x~$};
\node[scale=0.7] (b11) at (17,2.5)[draw=black, circle, fill=none]{};
\node[scale=0.7] () at (17.4,6.2)[draw=none, fill=none]{$t_x>..$};
\node[scale=0.7] (w21) at (12,5)[draw=black, circle, fill=none]{$s_x$};
\node[scale=0.7] () at (11.3,4)[draw=none, fill=none]{$t_x>..$};
\node[scale=0.7] (m11) at (14.5,2)[draw=black, circle, fill=none]{$t_x$};
\node[scale=0.7] () at (15.2,1)[draw=none, fill=none]{$s_x>x$};
\node[scale=0.7] (a11) at (20,6)[draw=none, fill=none]{};
\node[scale=0.7] (a21) at (20,4)[draw=none, fill=none]{};

\node[scale=0.7] (w12) at (29,5)[draw=black, circle, fill=none]{$~x~$};
\node[scale=0.7] (b12) at (29,2.5)[draw=black, circle, fill=none]{};
\node[scale=0.7] () at (29.4,6.2)[draw=none, fill=none]{$t_x>..$};
\node[scale=0.7] (w22) at (24,5)[draw=black, circle, fill=none]{$s_x$};
\node[scale=0.7] () at (23.3,4)[draw=none, fill=none]{$t_x>..$};
\node[scale=0.7] (m12) at (26.5,2)[draw=black, circle, fill=none]{$t_x$};
\node[scale=0.7] () at (27.2,1)[draw=none, fill=none]{$s_x>x$};
\node[scale=0.7] (a12) at (32,6)[draw=none, fill=none]{};
\node[scale=0.7] (a22) at (32,4)[draw=none, fill=none]{};

\node (h1) at (8,3) [draw=none,fill=none]{};
\node (h2) at (10.5,3) [draw=none,fill=none]{};
  
\node (h3) at (20,3) [draw=none,fill=none]{};
\node (h4) at (22.5,3) [draw=none,fill=none]{};
  
\path[->]
(h1) edge (h2)
(h3) edge (h4);

\path[dashed]
(w1) edge (w2)
(w11) edge (w21)
(w12) edge (w22)
(b1) edge (w1)
(b1) edge (m1)
(b11) edge (w11)
(b11) edge (m11)
(b12) edge (w12)
(b12) edge (m12);

\path[-, line width=3pt]
(w1) edge (m1)
(w21) edge (m11)
(w22) edge (m12)
(w12) edge (a22);

\path[-]
(w2) edge (m1)
(w11) edge (m11)
(w12) edge (m12)
(w1) edge (a1)
(w1) edge (a2)
(w11) edge (a11)
(w11) edge (a21)
(w12) edge (a12);

\node (a31) at (17,6)[draw=none, fill=none]{};
\node (a32) at (29,6)[draw=none, fill=none]{};

\node[overlay,cloud callout,callout relative pointer={(0.3cm,-0.45cm)},aspect=1.5,draw=black] at (14.5,8.5) {\small $t_x$};
\node[overlay,cloud callout,callout relative pointer={(0.3cm,-0.45cm)}, aspect=1.5,draw=black] at (26.5,8.5) {\small $t_x$};
\node () at (17,7.5)[draw=none, fill=none]{};
\end{tikzpicture}
\end{center}
\caption{\label{fig:memoryJammer}Effect of a memory jammer}
\end{figure}
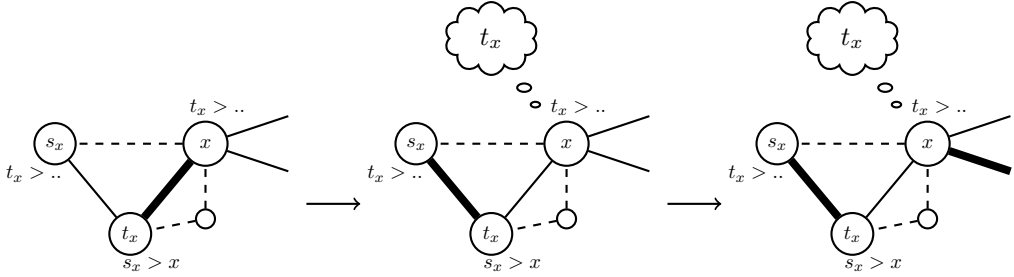

\begin{bspi}[Memory Jammer] 
\rm
Consider the graph given in Fig.~\ref{fig:memoryJammer}. As agents $x$ and $t_x$ are at distance 2 in the network $(V,L)$, edge $\{x,t_x\}$ is accessible, regardless of the current matching. Furthermore, it is the only accessible edge for $t_x$ when $t_x$ is unmatched. Let $t_x$ be the most preferred partner of $x$. Thus, $\{x,t_x\}$ is a local blocking pair as long as $t_x$ is unmatched. Once this blocking pair gets resolved, edge $\{t_x,s_x\}$ becomes accessible. Let $s_x$ be the most preferred partner of $t_x$, and $t_x$ the most preferred partner of $s_x$. Thus, $\{t_x,s_x\}$ becomes a local blocking pair, and $t_x$ will leave $x$ to match to $s_x$. Since $t_x$ is the most-liked partner, $x$ will keep it in his memory, and $t_x$ will never be replaced by any other agent. At the same time, $t_x$ and $s_x$  will never break their match regardless which other agents become accessible. Therefore, $x$ has no benefit from remembering $t_x$, and all subsequent sequences will evolve in similar fashion as without memory of $x$. \hfill $\blacksquare$
\end{bspi}

We augment each agent in the circling gadget with a separate memory jammer. In the following Proposition~\ref{prop1}, we show that this augmented circling gadget does not allow a sequence of local improvement steps from $M=\emptyset$ to any locally stable matching. We then show in Theorem~\ref{satzQuality} below that deciding this question is again \classNP-hard.

\begin{prop}
\label{prop1}
There is a bipartite network game with general preferences, quality memory and initial matching $M = \emptyset$ such that no locally stable matching can be reached with local improvement steps from $M$.
\end{prop}
\begin{proof}
As seen in the example above, we can add a memory jammer for every agent $A,B,C,1,2,3$ of the circling gadget to make sure that initially each of the ordinary agents $x$ gets to match with an agent $t_x$ which symbolizes its top choice. Then $t_x$ can (and will) switch to some agent $s_x$ and stay in this matching for the rest of the dynamics. At the same time $t_x$ will stay in the memory of gadget agent $x$ indefinitely, thereby making his memory useless. Recall that we argued above, by suitably ranking the $b_i$-agents at the bottom of every list, they are never involved in any local blocking pair other than $\{i,b_i\}$. As such, we do not need memory jammers for $b_i$-agents.

Formally, we consider
\begin{align*} 
 U=&\{1,2,3\}\cup\{t_x\mid x\in\{A,B,C\}\}\cup \{s_x\mid x\in\{1,2,3\}\}\\
 W=&\{A,B,C,b_1,b_2,b_3\}\cup\{d_x\mid x\in\{1,2,3,A,B,C\}\}\\
 &\cup \{t_x\mid x\in\{1,2,3\}\cup\{s_x\mid x\in\{A,B,C\}\} \\
 L=&\{\{A,b_1\},\{B,b_2\},\{C,b_3\},\{b_1,1\},\{b_2,2\},\{b_3,3\}\}\\
 &\cup\{\{A,B\},\{B,C\},\{C,A\}\}\\
 &\cup\{\{s_x,x\},\{x,d_x\},\{d_x,t_x\}\mid x\in\{1,2,3,A,B,C\}\} \mbox{ and}\\
 E=&\{\{u,v\}\mid u\in\{1,2,3\}, v\in\{A,B,C\}\}\cup\{\{x,t_x\},\{t_x,s_x\}\mid x\in \{1,2,3,A,B,C\}\}\enspace.
\end{align*}

The new preference lists and a sketch of the links can be found in Fig.~\ref{fig:qualityMemory}. Note that we restrict the possible matching edges to the set of edges previously considered in the circling gadget and described for the memory jammers. It will become obviou below that by suitably ranking agents at the bottom of the lists, we can also assume that $E=U \times W$.

\begin{figure}
\begin{center}
\begin{tabular}[b]{|l|l|}\hline
Agent & Preference List \\\hline\hline
$1 $ & $t_1\succ C\succ B\succ A$ \\\hline
$2 $ & $t_2\succ A\succ C\succ B$ \\\hline
$3 $ & $t_3\succ B\succ A\succ C$ \\\hline
$A $ & $t_A\succ 3\succ 1\succ 2$ \\\hline
$B $ & $t_B\succ 1\succ 2\succ 3$ \\\hline
$C $ & $t_C\succ 2\succ 3\succ 1$ \\\hline
$t_x$ & $s_x\succ x$ for $x\in \{1,2,3,A,B,C\}$\\\hline
$s_x$ & $t_x$ for $x\in \{1,2,3,A,B,C\}$\\\hline
\end{tabular} 
\hspace{0.5cm}
%
\begin{tikzpicture}[thick,scale=0.45, every node/.style={transform shape}]
\tikzstyle{vertex}=[circle,draw,minimum size=2pt]

\node[vertex] (1) at (2,9) {$1$};  
\node[vertex] (2) at (6,12) {$2$};
\node[vertex] (3) at (10,9) {$3$};
\node[vertex] (A) at (2,1) {$A$};
\node[vertex] (B) at (6,4) {$B$};
\node[vertex] (C) at (10,1) {$C$};
\node[vertex] (b1) at (2,5) {\tiny $b_1$}; 
\node[vertex] (b2) at (6,8) {\tiny $b_2$};
\node[vertex] (b3) at (10,5) {\tiny $b_3$};

\node[vertex] (d1) at (1,7) {\tiny $d_1$}; 
\node[vertex] (d2) at (5,10) {\tiny $d_2$};
\node[vertex] (d3) at (11,7) {\tiny $d_3$};
\node[vertex] (dA) at (1,-1) {\tiny $d_A$}; 
\node[vertex] (dB) at (5,2) {\tiny $d_B$};
\node[vertex] (dC) at (11,-1) {\tiny $d_C$};
\node[vertex] (t1) at (0,9) {$t_1$};  
\node[vertex] (t2) at (4,12) {$t_2$};
\node[vertex] (t3) at (12,9) {$t_3$};
\node[vertex] (tA) at (0,1) {$t_A$};
\node[vertex] (tB) at (4,4) {$t_B$};
\node[vertex] (tC) at (12,1) {$t_C$};
\node[vertex] (s1) at (1,11) {$s_1$};  
\node[vertex] (s2) at (5,14) {$s_2$};
\node[vertex] (s3) at (11,11) {$s_3$};
\node[vertex] (sA) at (1,3) {$s_A$};
\node[vertex] (sB) at (5,6) {$s_B$};
\node[vertex] (sC) at (11,3) {$s_C$};

\path[link] (1) -- (b1);
\path[link] (2) -- (b2);
\path[link] (3) -- (b3);
\path[link] (A) -- (b1);
\path[link] (B) -- (b2);
\path[link] (C) -- (b3);
\path[link] (A) -- (B);
\path[link] (B) -- (C);
\path[link] (C) -- (A);
\path[link] (s1) -- (1) -- (d1) -- (t1);
\path[link] (s2) -- (2) -- (d2) -- (t2);
\path[link] (s3) -- (3) -- (d3) -- (t3);
\path[link] (sA) -- (A) -- (dA) -- (tA);
\path[link] (sB) -- (B) -- (dB) -- (tB);
\path[link] (sC) -- (C) -- (dC) -- (tC);

\end{tikzpicture}
\vspace{0.3cm}
\end{center}
\caption{\label{fig:qualityMemory} Preferences and links in the circling gadget with memory jammers as described in Proposition~\ref{prop1}}
\end{figure}

To show that no locally stable matching can be reached, start with $M = \emptyset$ and assume for contradiction that there is a sequence of local improvement steps leading to a locally stable matching. A matching cannot be locally stable until $\{x,t_{x}\}$ was created once for all $x\in\{1,2,3,A,B,C\}$, because otherwise it is accessible, matches $x$ to his favorite partner, and matches $t_{x}$ to the only possible accessible partner when being single. Afterwards, every such edge will be replaced by $\{t_{x},s_{x}\}$, which remains stable since it matches both partners to their most preferred choice. Let $x^*$ be the last agent for which $\{t_{x^*},s_{x^*}\}$ is generated. At that moment $x^*$ is unmatched, and every agent of $\{1,2,3,A,B,C\}$ will continue to hold his $t$-partner in his memory. As the $t$-agents are not willing to change their matching edges, there will be no edge created from memory from this point on. If $x^*\in\{1,2,3\}$, one of the agents in $\{A,B,C\}$ is unmatched as well and vice versa. This leaves us in the situation described in the dynamics of the circling gadget. Since all memory entries are filled by $t_x$-agents, from this state no locally stable matching can be reached.
\end{proof}

\begin{bem} \rm
The same memory reset also works if every agent remembers the best $k$ previous matches, for any number $k$, simply by applying $k$ copies of the memory reset construction for each agent.
\end{bem}

\begin{satz}
\label{satzQuality}
It is \classNP-hard to decide \Reach\ to an arbitrary locally stable matching in a bipartite network game with quality memory.
\end{satz}

\begin{proof}
For the proof we use the \ThreeSAT\ gadget described in Corollary~\ref{satz4}, where we proved that no memory can help to reach a given stable matching $M$. We then combine this instance with augmented circling gadgets (similar as we did in Corollary~\ref{satzNPC}). In this way, we ensure that every stable matching must include the edges of $M$ to stabilize the gadgets. This matching, however, can be reached if and only if the \ThreeSAT\ formula is satisfiable.
\end{proof}

As a slightly stronger result, the same hardness result might be shown with the empty initial matching. For this, one might be able to adjust the construction of Theorem~\ref{satzReviewer} in order to avoid that quality memory helps reachability. We leave this issue as an avenue for future work.

\paragraph{Recency Memory}
In recency memory, every agent remembers the last agent it has been matched to before, which is yet another very natural update strategy. For our treatment we focus on the case in which the network links satisfy $L \subseteq (W \times W) \cup (U \times W)$. This is a natural restriction when one side consists of inanimate objects that form resources to be consumed. Theorem~\ref{satzReviewer} shows \classNP-hardness of deciding \Reach\ in this class of instances. In contrast, with recency memory it is always possible to reach a locally stable matching.

\begin{satz}\label{recency}
For every bipartite network game with general preferences, links $L \subseteq (U \times W) \cup (W \times W)$, recency memory and every initial matching, there is a sequence of $O(|U|^2|W|^2)$ many local improvement steps to a locally stable matching.
\end{satz}
\begin{proof}
Our basic approach is to construct the sequence in two phases similarly as in~\cite{AckermannGMRV11}. In the first phase, we let the matched agents from $U$ improve, but ignore the unmatched ones. In the second phase, we execute iterations, where in each iteration we pick a $u \in U$ and only do improvements steps involving $u$. Here we keep track of the agents from $W$ and ensure they improve at the end of each iteration.

\textit{Preparation phase}: As long as there is at least one $u \in U$ with $u$ matched and $u$ part of a local blocking pair, allow $u$ to switch to the better partner.

The preparation phase terminates after at most $|U|\cdot|W|$ steps, as in every round one matched $u \in U$ strictly improves in terms of preference. This can happen at most $|W|$ times for each matched $u$. In addition, the number of matched agents from $U$ can only decrease.

\textit{Memory phase}: As long as there is a $u \in U$ with $u$ part of a local blocking pair, pick $u$ and execute a sequence of local improvement steps involving $u$ until $u$ is not part of any local blocking pair anymore. For every edge $e = \{u',w\}$ with $u' \neq u$ that was deleted during the sequence, recreate $e$ from the memory of $u'$.\\

We claim that if we start the memory phase after the preparation phase, at the end of every iteration we have the following invariants: The agents of $W$ that have been matched before are still matched, they do not have a worse partner than before, and at least one of them is matched strictly better than before. Also, only unmatched agents from $U$ are involved in local blocking pairs.

Obviously, at the end of the preparation phase, all local blocking pairs contain unmatched agents of $U$, i.e., initially our invariant holds. Let $u$ be the agent chosen in the following iteration of the memory phase. At first, we consider the outcome for $w\in W$. If $w$ is the agent matched to $u$ in the end, then $w$ clearly has improved. Otherwise $w$ gets matched to its former partner (if it had one) through memory and thus is matched similarly as before. In particular, every $w$ that represents an improvement to some $u'$ but was blocked by a higher ranked agent, still remains blocked. Together with the fact that we execute local improvement steps involving $u$ until it is not part of a local blocking pair anymore, this guarantees that all matched agents of $U$ cannot improve at the end of this iteration of the memory phase. As one agent of $W$ strictly improves in every iteration, we have at most $|U|\cdot|W|$ iterations in the memory phase, where every iteration consists of at most $|W|$ steps by $u$ and at most $|U|-1$ edges reproduced from memory.
\end{proof}

The structure of $L$, combined with recency memory, plays a key role in this proof. During the memory phase, an agent $u \in U$ can only become unmatched if its current partner from $W$ deviates to a more desirable match. Due to the structure of $L$, agents in $U \setminus \{u\}$ and $W$ do not require a matching edge to $u$ in order to discover more desirable matches. Hence, it is not necessary to alter the memory of $u$ to stabilize the system with respect to another agent $u' \in U$. Due to recency memory, agent $u$ can reinsert the edge to the last partner, thereby allowing to recover the state of the system w.r.t.\ $u$ and make overall progress towards a locally stable matching. We illustrate the arguments with the following example.

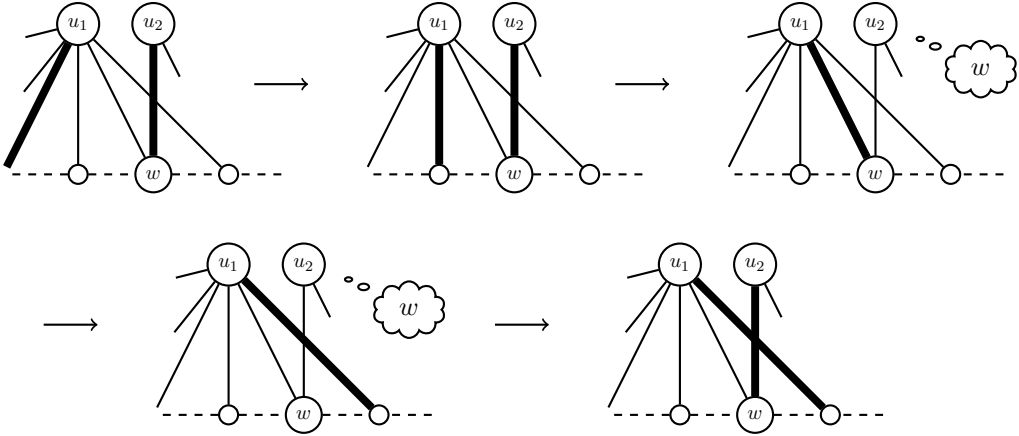
\begin{figure}
\begin{center}
\begin{tikzpicture}[thick,scale=0.4]
\node[scale=0.7] (w1) at (5,5)[draw=black, circle]{$u_2$};
\node[scale=0.7] (w2) at (2.5,5)[draw=black, circle]{$u_1$};
\node[scale=0.7] (m1) at (2.5,0)[draw=black, circle]{};
\node[scale=0.7] (m2) at (5,0)[draw=black, circle]{$w$};
\node[scale=0.7] (m3) at (7.5,0)[draw=black, circle]{};
\node[scale=0.7] (a1) at (0,0)[draw=none, fill=none]{};
\node[scale=0.7] (a2) at (9.5,0)[draw=none, fill=none]{};

\node[scale=0.7] (b1) at (6,3)[draw=none, fill=none]{};
\node[scale=0.7] (b2) at (0.5,4.5)[draw=none, fill=none]{};
\node[scale=0.7] (b3) at (0.5,2.5)[draw=none, fill=none]{};

\node[scale=0.7] (w11) at (17,5)[draw=black, circle]{$u_2$};
\node[scale=0.7] (w21) at (14.5,5)[draw=black, circle]{$u_1$};
\node[scale=0.7] (m11) at (14.5,0)[draw=black, circle]{};
\node[scale=0.7] (m21) at (17,0)[draw=black, circle]{$w$};
\node[scale=0.7] (m31) at (19.5,0)[draw=black, circle]{};
\node[scale=0.7] (a11) at (12,0)[draw=none, fill=none]{};
\node[scale=0.7] (a21) at (21.5,0)[draw=none, fill=none]{};

\node[scale=0.7] (b11) at (18,3)[draw=none, fill=none]{};
\node[scale=0.7] (b21) at (12.5,4.5)[draw=none, fill=none]{};
\node[scale=0.7] (b31) at (12.5,2.5)[draw=none, fill=none]{};

\node[scale=0.7] (w12) at (29,5)[draw=black, circle]{$u_2$};
\node[scale=0.7] (w22) at (26.5,5)[draw=black, circle]{$u_1$};
\node[scale=0.7] (m12) at (26.5,0)[draw=black, circle]{};
\node[scale=0.7] (m22) at (29,0)[draw=black, circle]{$w$};
\node[scale=0.7] (m32) at (31.5,0)[draw=black, circle]{};
\node[scale=0.7] (a12) at (24,0)[draw=none, fill=none]{};
\node[scale=0.7] (a22) at (33.5,0)[draw=none, fill=none]{};

\node[scale=0.7] (b12) at (30,3)[draw=none, fill=none]{};
\node[scale=0.7] (b22) at (24.5,4.5)[draw=none, fill=none]{};
\node[scale=0.7] (b32) at (24.5,2.5)[draw=none, fill=none]{};

\node[scale=0.7] (w13) at (10,-3)[draw=black, circle]{$u_2$};
\node[scale=0.7] (w23) at (7.5,-3)[draw=black, circle]{$u_1$};
\node[scale=0.7] (m13) at (7.5,-8)[draw=black, circle]{};
\node[scale=0.7] (m23) at (10,-8)[draw=black, circle]{$w$};
\node[scale=0.7] (m33) at (12.5,-8)[draw=black, circle]{};
\node[scale=0.7] (a13) at (5,-8)[draw=none, fill=none]{};
\node[scale=0.7] (a23) at (14.5,-8)[draw=none, fill=none]{};

\node[scale=0.7] (b13) at (11,-5)[draw=none, fill=none]{};
\node[scale=0.7] (b23) at (5.5,-3.5)[draw=none, fill=none]{};
\node[scale=0.7] (b33) at (5.5,-5.5)[draw=none, fill=none]{};

\node[scale=0.7] (w14) at (25,-3)[draw=black, circle]{$u_2$};
\node[scale=0.7] (w24) at (22.5,-3)[draw=black, circle]{$u_1$};
\node[scale=0.7] (m14) at (22.5,-8)[draw=black, circle]{};
\node[scale=0.7] (m24) at (25,-8)[draw=black, circle]{$w$};
\node[scale=0.7] (m34) at (27.5,-8)[draw=black, circle]{};
\node[scale=0.7] (a14) at (20,-8)[draw=none, fill=none]{};
\node[scale=0.7] (a24) at (29.5,-8)[draw=none, fill=none]{};

\node[scale=0.7] (b14) at (26,-5)[draw=none, fill=none]{};
\node[scale=0.7] (b24) at (20.5,-3.5)[draw=none, fill=none]{};
\node[scale=0.7] (b34) at (20.5,-5.5)[draw=none, fill=none]{};

\node (h1) at (8,3) [draw=none,fill=none]{};
\node (h2) at (10.5,3) [draw=none,fill=none]{};
  
\node (h3) at (20,3) [draw=none,fill=none]{};
\node (h4) at (22.5,3) [draw=none,fill=none]{};

\node (h5) at (1,-5) [draw=none,fill=none]{};
\node (h6) at (3.5,-5) [draw=none,fill=none]{};
  
\node (h7) at (16,-5) [draw=none,fill=none]{};
\node (h8) at (18.5,-5) [draw=none,fill=none]{};

\path[->]
(h1) edge (h2)
(h3) edge (h4)
(h5) edge (h6)
(h7) edge (h8);

\path[link]
(m1) edge (m2)
(m2) edge (m3)
(m11) edge (m21)
(m21) edge (m31)
(m12) edge (m22)
(m22) edge (m32)
(m13) edge (m23)
(m23) edge (m33)
(m14) edge (m24)
(m24) edge (m34)
(m1) edge (a1)
(a2) edge (m3)
(m11) edge (a11)
(a21) edge (m31)
(m12) edge (a12)
(a22) edge (m32)
(m13) edge (a13)
(a23) edge (m33)
(m14) edge (a14)
(a24) edge (m34);

\path[edge]
(b1) edge (w1)
(b2) edge (w2)
(b3) edge (w2)
(b11) edge (w11)
(b21) edge (w21)
(b31) edge (w21)
(b12) edge (w12)
(b22) edge (w22)
(b32) edge (w22)
(b13) edge (w13)
(b23) edge (w23)
(b33) edge (w23)
(b14) edge (w14)
(b24) edge (w24)
(b34) edge (w24);

\path[existEdge]
(w1) edge (m2)
(w11) edge (m21)
(w14) edge (m24);

\path[edge]
(w12) edge (m22)
(w13) edge (m23);

\path[existEdge]
(w2) edge (a1);

\path[edge]
(w21) edge (a11)
(w22) edge (a12)
(w23) edge (a13)
(w24) edge (a14);

\path[existEdge]
(w21) edge (m11);

\path[edge]
(w2) edge (m1)
(w22) edge (m12)
(w23) edge (m13)
(w24) edge (m14);

\path[existEdge]
(w22) edge (m22);

\path[edge]
(w2) edge (m2)
(w21) edge (m21)
(w23) edge (m23)
(w24) edge (m24);

\path[existEdge]
(w23) edge (m33)
(w24) edge (m34);

\path[edge]
(w2) edge (m3)
(w21) edge (m31)
(w22) edge (m32);

\node (a32) at (30,4)[draw=none, fill=none]{};
\node (a33) at (11,-4)[draw=none, fill=none]{};

\node[overlay,cloud callout,callout relative pointer={(-0.4cm,0.2cm)},%
        aspect=1.5,draw=black] at (32.5,3.5) {\small $w$};
\node[overlay,cloud callout,callout relative pointer={(-0.4cm,0.2cm)},%
        aspect=1.5,draw=black] at (13.5,-4.5) {\small $w$};
\end{tikzpicture}
\end{center}
\caption{\label{fig:recencyMemory}Example of a sequence computed by the algorithm in Theorem~\ref{recency}}
\end{figure}

\begin{bspi}~
\rm
In the example given in Figure~\ref{fig:recencyMemory} we execute an iteration of the memory phase. $u_1 \in U$ is the agent currently chosen. During the phase, $u_1$ repeatedly switches partner to improve. In that process, it also matches to agent $w$. $w$ was matched to $u_2$ at the beginning of the iteration. However, $w$ is only an intermediate partner for $u_1$, who can improve further once new agents become accessible via the links of $w$. Thus, at the end of the iteration when $u_1$ is involved in no further local blocking pairs, the edge $\{u_2,w\}$ can be retrieved from $u_2$'s memory.\hfill $\blacksquare$
\end{bspi}

\begin{bem} \rm
Note that the network used in Proposition~\ref{prop1} does not fulfill the criteria of $L \subseteq (U \times W) \cup (W \times W)$, and thereby the impossibility results for quality memory do not apply. However, if we discard all $d$-agents with their links as well as all $\{x,s_x\}$-links, install direct links between $t_x$ and $s_x$ and use $\{\{x,t_x\}\mid x\in \{1,2,3\}\cup\{A,B,C\}\}$ as initial state, the memory gets jammed in the same way, and the social network now is limited to $L \subseteq (U \times W) \cup (W \times W)$. Theorem~\ref{satzQuality} easily continues to hold also with this type of circling gadget, which proves that for quality memory \Reach\ remains \classNP-hard even in instances with $L \subseteq (U \times W) \cup (W \times W)$.
\end{bem}

The following is a direct corollary from the previous theorem. In random dynamics, with probability 1 in the limit, we will at least once start in a matching and execute the sequence described in the last theorem. 

\begin{fol}
For every bipartite network game with general preferences, links $L \subseteq (U \times W) \cup (W \times W)$, recency memory, and every initial matching, random dynamics converge to a locally stable matching in the limit with probability 1.
\end{fol}

Sadly, we cannot always expect fast convergence here, as there are instances where random dynamics yield a sequence of exponential length with high probability even in ordinary stable marriage where all pairs of agents are accessible (see~\cite{AckermannGMRV11}).

Our constructions require only the agents from $U$ to have recency memory, which allows us to show the same results when we have no memory for agents in $W$. In contrast, if we omit the memory for $U$ and just allow recency memory for $W$, no sequence of local improvement steps might lead to a locally stable matching. 
\begin{prop}
\label{prop2}
There is a bipartite network game with general preferences, links $L \subseteq (U \times W) \cup (W \times W)$, recency memory for agents in $W$ and initial matching $M = \emptyset$ such that no locally stable matching can be reached with local improvement steps from $M$.
\end{prop}
\begin{proof}
Consider the circling gadget. First, as noted before, memory for the $b_i$-agents has no effect since the only matching edges that can evolve are $\{i,b_i\}$. All agents of $W$ are matched in every locally stable matching, and none of the edges present in any locally stable matching is accessible via two links. As such, the last edge to complete any locally stable matching must come from memory. Now $\{1,A\},$ $\{2,B\},$ $\{3,C\}$ are always accessible. $\{1,C\},$ $\{2,A\},$ and $\{3,B\}$ only get removed if the $W$-agent finds a better partner. This shows that the locally stable matching $\{\{1,C\},$ $\{2,A\},$ $\{3,B\}\}$ cannot be reached, even with recency memory. By symmetry, assume the last edge recovered from memory be $\{2,A\}$. Then $2$ must be the last partner $A$ was matched to. Since $A$ must leave $2$ for a better partner (for $1$), and $\{1,A\}$ must become deleted subsequently, $A$ removes $2$ from its memory.

As such, this leaves us with the locally stable matching $M^*=\{\{1,B\}$, $\{2,C\},$ $\{3,A\}\}$. Suppose the last edge of the matching created from memory is $\{3,A\}$. Then $\{2,B\}$ must have been created from the memory of $B$, since otherwise it needs $\{1,A\}$ to exist, which implies the memory of $A$ is not set ot $3$. For the same reason, $\{2,C\}$ must have been created from memory, since for symmetric reasons $\{1,B\}$ had not been able to be created from memory. Hence, in order to reach $\{\{1,B\}$, $\{2,C\},$ $\{3,A\}\}$ we must create all three edges from memory. Now consider for agent $A$ the earliest step where the memory is set to $3$ and remains set like this until the end. We say the memory for $A$ is finalized in this step. Similarly, we consider the steps where the memories of agents $B$ and $C$ are finalized. Assume w.l.o.g.\ that the memory of $A$ is finalized last. Then edge $\{3,A\}$ must be constructed by existence of $\{3,C\}$. This, however, implies that the memory for $C$ is finalized later, a contradiction.
\end{proof}
The proposition can be used with the construction of Corollary~\ref{satz4} in the same manner as Proposition~\ref{prop1} was used to show Theorem~\ref{satzQuality}. This yields the following result. As for quality memory, a similar result might be possible from the empty initial matching. We leave this issue as an open problem.
\begin{satz}
\label{satzRecency}
It is \classNP-hard to decide \Reach\ to an arbitrary locally stable matching in a bipartite network game when recency memory exists only for one partition.
\end{satz}


\paragraph{Random Memory}
\label{sec:random}

Let us now consider how random memory might help us reach a locally stable matching from every starting state even in general bipartite network games. Again, we cannot expect fast convergence due to the full-information lower bound in~\cite{AckermannGMRV11}. However, we show that random memory can help with reachability:

\begin{satz}
For every bipartite network game with random memory, random dynamics converge to a locally stable matching with probability 1. 
\end{satz}

\begin{proof}
Our proof combines an idea used in~\cite{Hoefer13} with a convergence result in~\cite{RothVV90}. We consider a sequence of random local improvement steps and divide it into phases. Whenever a new edge is created for the first time, a new phase starts. Hence, phase $t$ contains the part of the sequence, where exactly $t$ edges have been created for the first time. Within a phase we have a fixed number of matching edges available (from the network or from memory). Consider phase $t^*$ where $t^*$ is the maximal phase of the sequence. $t^*$ exists, as $t$ is monotonically increasing and limited by $|E|$. We know that every phase $t < t^*$ ends after a finite number of steps. Hence we only have to show that phase $t^*$ is finite in expectation. Roth and Vande Vate~\cite{RothVV90} demonstrate how to construct a sequence of blocking pair resolutions to form a stable matching with full information when all possible matching edges are accessible. In phase $t^*$, we have $t^*$ edges that can be used for the matching, and all of them can also be remembered. Thus, with non-zero probability there is an initial state such that the random memory remembers the local blocking pairs from the sequence in the correct order and the random dynamics implement the local blocking pair resolutions in the correct way. Thus, phase $t^*$ is ends after a finite number of steps in expectation. This proves the theorem.
\end{proof}

\section{Centralized problems}
\label{sec:central}

\subsection{Maximum locally stable matching}
\label{sec:IS}

The size of locally stable matchings in an instance can be highly variable -- for example, in the trivial example where there are no links, the empty matching as well as a perfect matching are locally stable. From a designer's perspective, since all agents strive to be matched, it is desirable to form locally stable matchings of maximal size. Unfortunately, we will show that there is a close connection between maximum independent sets and maximum locally stable matchings. This allows us to transfer hardness of approximation results for independent set to locally stable matching. 

\begin{satz}
\label{satz9}
For every graph $G$ we can build a network $N$ such that $N$ holds a maximum locally stable matching of size $|V\left[G\right]|+k$ iff $G$ has a maximum independent set of size $k$.
\end{satz}

\begin{proof}
Consider a graph $G=(V,E)$, $|V|=n$, as an instance of the maximum independent set problem. We construct the following job-market network $N=(V'=U\cup W,L)$. For every $v\in V$ we have agents $u_{v,1}$ and $u_{v,2}$ in $U$ and agents $w_{v,1}$ and $w_{v,2}$ in $W$. Further there are links $\{w_{v,1},w_{v',2}\}$ and $\{w_{v',2},w_{v,2}\}$, if $v'\in N(v)$. We allow matching edges $\{u_{v,1},w_{v,1}\}$, $\{u_{v,1},w_{v',2}\}$ for $v'\in N(v)$, $\{u_{v,1},w_{v,2}\}$ and $\{u_{v,2},w_{v,2}\}$. Each $u_{v,1}$ prefers $w_{v,2}$ to every $w_{v',2}$, $v'\in N(v)$, and every $w_{v',2}$ to $w_{v,1}$. The preferences between the the different neighbors can be chosen arbitrarily. Each $w_{v,2}$ prefers $u_{v,1}$ to every $u_{v',1}$, $v'\in N(v)$, and every $u_{v',2}$ to $u_{v,2}$. Again the neighbors can be ordered arbitrarily. The agents $w_{v,1}$ and $u_{v,2}$ have only one possible matching partner anyway.

We claim that $G$ has a maximum independent set of size $k$ iff $N$ has a locally stable matching of size $n+k$. First, let $S$ be a maximum independent set in $G$. Then $M=\{\{u_{v,1},w_{v,2}\}\mid v\in V\setminus S\}\cup\{\{u_{v,1},w_{v,1}\},$ $\{u_{v,2},w_{v,2}\}$ $\mid v\in S\}$ is a locally stable matching in $N$. From edges $\{u_{v,1},w_{v,2}\}$ no agent wants to deviate. For the other agents, the independent set property tells us that for $v\in S$ all agents $v'\in N(S)$ yield edges $\{u_{v',1},w_{v',2}\}$ that keep $u_{v,1}$ from switching to $w_{v',2}$. Thus, from $\{u_{v,1},w_{v,1}\}$ no agents wants to deviate. $w_{v,2}$ is not accessible with $u_{v,1}$, which implies that the agents do not want to deviate from $\{u_{v,2},w_{v,2}\}$.

Now, let $M$ be a maximum locally stable matching for $N$. Further, we chose $M$ such that every $u_{v,1}$ is matched, which is possible as replacing a matching partner of $w_{v,2}$ by (the unmatched) $u_{v,1}$ will not generate local blocking pairs or lower the size of $M$. We note that no $u_{v,1}$ is matched to some $w_{v',2}$ with $v\neq v'$, since this would imply $\{u_{v,1}, w_{v,2}\}$ becomes accessible and a local blocking pair. Then for $S=\{v\mid u_{v,2}\in M\}$, $|S|=|M|-n$. Every $u_{v,2}$ can only be matched to $w_{v,2}$, which means that $u_{v,1}$ must be matched to $w_{v,1}$. These agents must not be involved in local blocking pairs, which implies for every $v\in S$, $N(v)\cap S=\emptyset$. Thus, $S$ is an independent set.
\end{proof}

\begin{fol}
Under the unique games conjuncture maximum locally stable matching cannot be approximated within $1.5-\epsilon$, where $\epsilon$ approaches 0 when $n$ grows.
\end{fol}

\begin{proof}
We combine the relation to maximum independent set as shown in Theorem~\ref{satz9} with a result of~\cite{AustrinKS11} that independent set is unique-games-hard to approximate within a factor of $\Omega\left(\frac{d}{\log^2(d)}\right)$ for independent sets of size $k= \left(\frac{1}{2}-\Theta\left(\frac{\log(\log(d))}{\log(d)}\right)\right)n$, where $d$ is the maximum degree. Hence, it is hard to distinguish instances with maximum independent set of size $k$ and $k\cdot O\left(\frac{\log^2(d)}{d}\right)$. With $d=\Theta(n)$, maximum locally stable matching is unique-games-hard to approximate within
\[
\frac{n+k}{n+k\cdot O\left(\frac{\log^2(n)}{n}\right)} \; = \; \frac{n + \left(\frac{1}{2}-\Theta\left(\frac{\log(\log(n))}{\log(n)}\right)\right)n}{n + \left(\frac{1}{2}-\Theta\left(\frac{\log(\log(n))}{\log(n)}\right)\right)n\cdot O\left(\frac{\log^2(n)}{n}\right)} \; \leq \; 1.5-\epsilon\enspace,\]
for any constant $\epsilon > 0$.
\end{proof}

\begin{bem} \rm
Note that in fact we only used the setting of the bipartite job-market game, where one side has no network at all. This shows that even under quite strong restrictions the hardness of approximation holds.
\end{bem}

\begin{prop}
If a stable matching exists, every stable matching is a 2-approximation of a maximum locally stable matching.
\end{prop}
\begin{proof}
Note that every stable matching is locally stable as well. Now let $M$ be a stable matching and $e=\{u,v\}$ an edge of a maximum locally stable matching $M^*$. We show that at least one agent of $e$ is matched in $M$. Then obviously $|M^*| < 2|M|$. Assume that both agents are unmatched in $M$. As $e$ exists in $M^*$, $u$ and $v$ prefer each other to being alone. Thus $\{u,v\}$ is a blocking-pair and $M$ cannot be stable.
\end{proof}

\subsection{Roommates Problem}
\label{sec:roommate}

If $E$ is bipartite, there always exists a stable (and thus, a locally stable) matching. In the more general roommates case, there are instances with general preferences such that no stable matching exists. The same obviously holds for locally stable matchings, since locality has no effect if $L$ contains a link for every pair of agents. While the existence of a stable matching can be decided in polynomial time~\cite{Irving85}, we show that the same question is \classNP-hard for locally stable matchings. Our initial proof was significantly simplified by an anonymous referee, who kindly provided the following proof.

\begin{satz}
\label{thm:roommates}
It is \classNP-complete to decide if a locally stable matching exists in a network game.
\end{satz}
\begin{proof}
Let \SRT\ denote the problem of deciding whether an instance of the stable roommates problem with ties admits a stable matching. \SRT\ is \classNP-complete even if each agent finds all other agents acceptable, and each tie is of length 2 and occurs at the head of some agent's preference list \cite{IrvingM02}.

We will use a reduction from this restricted version of \SRT\ to show that our problem is \classNP-complete. Let $I$ be an instance of the stated restriction of \SRT, and $V$ is the set of agents in $I$. Let $V' \subseteq V$ denote the set of agents in $I$ who have a tie (of length 2) at the head of their preference list in $I$. We form an instance $J$ of \LocalSRT\ by letting the potential matching edges $E$ be the edge set of the complete graph on $V$. Initially, let $L = E$, i.e., the links comprise all potential matching edges.

Each agent $v \in V$ initially has the same preference list in $J$ as in $I$. If $v \in V'$, then assume that the tie at the head of $v$'s preference list involves agents $v'$ and $v''$. By inspection of the proofs of Theorem~6.1 and Corollary~6.3 in \cite{IrvingM02}, for one of $v'$ and $v''$ (w.l.o.g.\ let that be $v''$) the unique first partner on the preference list is $v$. Moreover, $v' \not\in V'$ and $v'' \not\in V'$. Break the tie in $J$ so that $v$ prefers $v'$ to $v''$, and then $v''$ is preferred to all other agents on $v$'s preference list, as before. Remove all links incident to $v$ from $L$ apart from the link $\{v, v''\}$. Moreover, remove the link $\{v', v''\}$ from $L$. Hence, now if $v$ is unmatched or matched to $v''$, it is not accessible with $v'$.

We claim that $I$ has a stable matching if and only if $J$ has a locally stable matching. First, let $M$ be a stable matching in $I$. Suppose that $\{v_1, v_2\}$ is a local blocking pair of $M$ in $J$. If $v_1 \not\in V'$ and $v_2 \not\in V'$ then it is clear that $\{v_1, v_2\}$ is a blocking pair of $M$ in $I$, a contradiction. Now suppose w.l.o.g.\ let $v_1 \in V'$. Let $v_1'$ and $v_1''$ be the two members of $v_1$'s tie in $I$, and w.l.o.g.\ suppose $v_1$ prefers $v_1'$ to $v_1''$ in $J$. By inspection of the proofs of Theorem~6.1 and Corollary~6.3 in~\cite{IrvingM02}, it may be verified that $M(v_1) \in \{v_1', v_1''\}$. It follows that $v_2 = v_1'$ and $M(v_1) = v_1''$. But by construction of $L$, $v_1$ and $v_1'$ are not accessible, a contradiction. Hence $M$ is locally stable in $J$.

Conversely, suppose that $M$ is a locally stable matching in $J$. Suppose that $\{v_1, v_2\}$ is a blocking pair of $M$ in $I$. If $\{v_1, v_2\} \in L$ then $\{v_1, v_2\}$ is a local blocking pair of $M$ in $J$, a contradiction. Next suppose that $\{v_1, v_2\} \not\in L$. 
If $v_1 \not\in V'$ and $v_2 \not\in V'$, then $v_1$ and $v_2$ are the two members of the tie of length 2 in some agent $v$'s preference list. W.l.o.g.\ suppose that $v$ prefers $v_1$ to $v_2$ in $J$. Note $\{v, v_2\} \in L$ and $v_2$ ranks $v$ in first place, so $M(v) \in \{v_1, v_2\}$, for otherwise $\{v, v_2\}$ is a local blocking pair of $M$, a contradiction. By inspection of the proofs of Theorem~6.1 and Corollary~6.3 in~\cite{IrvingM02}, if $\{v, v_1\} \in M$, then $v_1$ prefers $v$ to $v_2$ in $I$, whilst if $\{v, v_2\} \in M$, $v_2$ prefers $v$ to $v_1$ in $I$. Hence $\{v_1, v_2\}$ cannot be a blocking pair of $M$ in $I$ after all. Finally, suppose $\{v_1, v_2\} \not\in L$ and w.l.o.g.\ $v_1 \in V'$ and $v_2 \not\in V'$. Let $v_1'$ and $v_1''$ be the two members of $v_1$'s tie in $I$. If $M(v_1) \in \{v_1',v_1''\}$, then $v_1$ cannot be involved in any blocking pair in $I$. Otherwise, $v_1$ prefers each of $v_1'$ and $v_1''$ to $M(v_1)$ in $J$. By construction of $J$, $\{v_1, v_1''\} \in L$ and $v_1''$ ranks $v_1$ in first place. Hence $\{v_1, v_1''\}$ is a local blocking pair of $M$ in $J$, a contradiction. Thus $M$ is stable in $I$.
\end{proof}

\section{Conclusion}

In this paper, we study \Reach\ problems for locally stable matchings. Our results show that locally stable matchings might not necessarily be reachable by a sequence of local improvement steps. We prove that deciding this question is \classNP-hard in many rather restricted domains of two-sided instances. Moreover, even in cases where a locally stable matching can be reached from an initial matching, the length of the shortest such sequence can be exponentially long. An interesting case, where locally stable matchings can always be reached, are instances with correlated preferences~\cite{Hoefer13}. For these instances we show that, in fact, every reachable matching can be reached within $O(|E^3|)$ many steps. Another interesting and natural property, which can be used to overcome these negative conditions, is memory. With recency memory we can guarantee existence of a sequence of $O(|U|^2|W|^2)$ many steps to a locally stable matching when there is no network among one partition. With random memory convergence is guaranteed in the limit with probability 1. 

In this direction, there are a variety of open problems. It is not known if the polynomial bounds we provide for correlated preferences or recency memory are tight. Furthermore, it would be interesting to see if there are other meaningful classes of instances, for which existence of a path to stability (of polynomial length) can be shown, or at least decided in polynomial time.

In addition to convergence questions, we study existence and optimization problems regarding locally stable matchings. For finding the locally stable matching of maximum cardinality, we provide a lower bound of $1.5-\epsilon$ on the approximation factor under the unique games conjecture, for every constant $\epsilon$. In the roommates case, even existence of locally stable matchings is shown to be \classNP-complete to decide. 

Perhaps the most interesting open problem in this domain is whether there is an efficient 1.5-approximation algorithm for finding a maximum locally stable matching. Alternatively, can the lower bound can be strengthened towards the simple upper bound of 2? For the roommates case, it would be interesting to identify further classes of instances, for which existence of a locally stable matching is guaranteed or at least decidable in polynomial time.

\subsubsection*{Acknowledgement}
We thank the anonymous reviewers for many helpful suggestions. In particular, we thank a reviewer that formulated and provided the significantly simplified and slightly more general proofs presented above for Theorems~\ref{satzReviewer} and~\ref{thm:roommates}.

\bibliographystyle{plain}
\bibliography{../../../../Bibfiles/literature,../../../../Bibfiles/martin}{}

\end{document}